\documentclass[conference]{IEEEtran}
\IEEEoverridecommandlockouts
\usepackage{subcaption}	

\usepackage{cite}
\usepackage{mathtools}
\usepackage{algorithmic}
\usepackage{graphicx}
\usepackage{textcomp}
\usepackage{xcolor}
\def\BibTeX{{\rm B\kern-.05em{\sc i\kern-.025em b}\kern-.08em
    T\kern-.1667em\lower.7ex\hbox{E}\kern-.125emX}}

\usepackage{graphicx}
\usepackage{booktabs} % For formal tables
\usepackage{markdown}
\usepackage{xcolor}
\usepackage{siunitx}
\usepackage{multirow}
\usepackage[utf8]{inputenc}
\usepackage{longtable}
\usepackage{afterpage}
\usepackage{enumitem}
\usepackage{xspace}
\usepackage[linesnumbered]{algorithm2e}
\usepackage{amsmath}
\usepackage{listings}
\usepackage{textcomp}
\usepackage{color}

\usepackage{amsthm}

\usepackage{hyperref}
\hypersetup{pdfborderstyle={/S/U/W 1}}

\newif\iftechreport
\techreportfalse

\numberwithin{equation}{section}

\newtheorem*{theorem}{Theorem}

\DeclareMathOperator*{\argmax}{\arg \max}

\DeclareMathOperator*{\E}{E}
\DeclareMathOperator*{\Var}{Var}
\newcommand{\Ne}[0]{N^1}

\let\vec\mathbf
\let\mat\mathbf

\newcommand{\sys}{ExSample\xspace}
\newcommand{\todo}[1]{\textcolor{red}{TODO:} \textcolor{gray}{#1}}
\newcommand{\TODO}[1]{\todo{#1}}
\newcommand{\srm}[1]{\textcolor{red}{Sam: #1}}
\newcommand{\favyen}[1]{\textcolor{red}{Favyen: #1}}

\renewcommand{\favyen}[1]{}
\renewcommand{\todo}[1]{}
\renewcommand{\TODO}[1]{}
\renewcommand{\srm}[1]{\xspace}

\makeatletter
\newcommand{\hlinelabel}[1]{\linelabel{#1}\Hy@raisedlink{\hypertarget{#1}{}}}
\makeatother
\renewcommand{\hlinelabel}[1]{}

\newcommand{\hlineref}[1]{\hyperlink{#1}{\autopageref{#1}~ln.~\lineref{#1}}}

\newcommand{\lsb}[0]    {{\tt random+}}
\newcommand{\random}[0] {{\tt random}}

\newcommand{\blazeit}[0]{{BlazeIt\xspace}}

\newenvironment{highlightenv}{\color{blue}}{\ignorespacesafterend}
\renewenvironment{highlightenv}{}{\ignorespacesafterend}
\newcommand{\hl}[1]{#1}

\newcommand{\highlight}[1]{\textcolor{blue}{#1}}
\renewcommand{\highlight}[1]{#1}

\definecolor{codegreen}{rgb}{0,0.6,0}
\definecolor{codegray}{rgb}{0.5,0.5,0.5}
\definecolor{codepurple}{HTML}{C42043}
\definecolor{backcolour}{HTML}{FFFFFF}
\definecolor{bookColor}{cmyk}{0,0,0,0.90}  
\color{bookColor}

\lstdefinestyle{mystyle}{
    backgroundcolor=\color{backcolour},   
    commentstyle=\color{codegreen},
    keywordstyle=\color{codepurple},
    numberstyle=\numberstyle,
    stringstyle=\color{codepurple},
    basicstyle=\footnotesize\ttfamily,
    breakatwhitespace=true,
    breaklines=true,
    captionpos=b,
    keepspaces=true,
    deletekeywords={IDENTITY},
    morekeywords={clustered, box, string, polygon},    
    framesep=1pt,
    framextopmargin=1pt,
    framexbottommargin=1pt,
    numbers=none, % could be eg. left
    numbersep=10pt,
    showspaces=false,
    showstringspaces=false,
    showtabs=false,
}
\newcommand\numberstyle[1]{%
    \footnotesize
    \color{codegray}%
    \ttfamily
    \ifnum#1<10 0\fi#1 |%
}

\makeatletter
\newcommand{\linebreakand}{%
  \end{@IEEEauthorhalign}
  \hfill\mbox{}\par
  \mbox{}\hfill\begin{@IEEEauthorhalign}
}
\makeatother

\setlist[itemize]{leftmargin=*,nosep}
\setlist[enumerate]{leftmargin=*,nosep}

\if{0}
\vldbTitle{xxx}
\vldbAuthors{xxx}
\vldbDOI{https://doi.org/10.14778/xxxxxxx.xxxxxxx}
\vldbVolume{21}
\vldbNumber{xxx}
\vldbYear{2021}
\fi

\newcommand{\reviewcomment}[1]{{\textcolor{red} {\textbf{\underline{#1}}}}}

\begin{document}

\if{0} % reviewer comments
\twocolumn 
\section*{Response to Reviewer Comments}
% https://cmt3.research.microsoft.com/ICDE2022/Submission/Reviews/558

We thank the reviewers for their constructive comments. We have highlighted our major changes \highlight{in blue}.

% 2. Area Chair's Summary
% The reviewers have identified a number of items regarding this manuscript that need to be improved, particularly the experiments.
% 3. Required Changes
% Please address all the opportunities for improvement identified by the reviewers as essential.
% 4. Optional Changes
% Please address all other opportunities for improvement identified by the reviewers.

\noindent\rule{\columnwidth}{1pt}

% 7. Describe in detail all strong points, labeled S1, S2, S3, etc.
% S1 - ExSample increases the speed in solving queries up to 4x compared to random sampling
% S2 - ExSample, compared to proxy model approaches where every frames are scored, reduces the cost/processing time
% S3 - A third party tuning optimization ( orthogonal to ExSample) can be applied in conjunction to this solution.
% 8. Describe in detail all opportunities for improvement, labeled O1, O2, O3, etc.
% 01 - It is not clear how number of chunks in which the input video is split are chosen
% 02 - The generation of the ground truth is approximate and the authors do not use benchmarks to compare the results
% 03- In the experimental results there is no plot showing the real gains in response time between ExSample and competitors
% 04- Future work is missing
% 9. Please rank the three most critical strong points and/or opportunities for improvement that led to your overall evaluation rating, e.g., "{S1, O7, O4}"
% S1,S2,O1,O3
% 10. Inclusive writing: does the paper follow the ICDE inclusive writing guidelines with respect to language, examples, figures? Please provide suggestions to the authors on how to improve the writing of the paper according to the guidelines.
% The paper follows the ICDE inclusive writing guidelines
% 14. List required changes for a revision, if appropriate, by identifying subset of previously specified opportunities for improvement (e.g., O1, O3, O6).
% O1,O3,O4

\reviewcomment{R3O1} {\em It is not clear how number of chunks in which the video is split is chosen}.  For our benchmarks we split videos into 20 minute chunks. Some videos are shorter than 20 minutes, in which case we use the whole video as a single chunk; for example, BDD consists of 1000 2-minute clips: each of these is used as a chunk. We have changed the description of our datasets in \hlineref{line:datasets} to mention this. Finally, if this question is more broadly about about how the number of chunks should be chosen, we address this question in our response to R402 below.\\
\reviewcomment{R3O2} {\em The generation of the ground truth is approximate and the authors do not use benchmarks to compare the results.} We originally evaluated on 5 datasets: three datasets (archie, amsterdam, and night-street) are existing benchmarks from NoScope~\cite{noscope} that are widely used in existing video analytics work~\cite{focus,blazeit,chameleon}; one dataset (BDD 1k) is an object detection benchmark in the machine learning community; and the last dataset (dashcam) is one we collected ourselves. We used these datasets to make the evaluation as compatible with existing work as possible, though the existing work we looked at does not address the distinct object problem when returning results from video. There is no more widely used benchmark at the moment, which we would have used if there was.  

Regarding the ground truth: the archie, amsterdam, and night-street datasets from prior work use approximate ground truth generated by an off-the-shelf object detector, and we adopt the same approach for BDD and dashcam. However, when evaluating the existing ground truth for those datasets, we found excessive `flickering' of detections over time, which causes an artificially high number of distinct instances, which in reality were simply disconnected segments of a single instance. To reduce this problem, and to substantially improve the accuracy of our ground truth labels over prior work (especially on the three datasets from NoScope), we manually labelled hundreds of frames and fine-tuned detectors for each dataset, and then used these detectors to generate improved ground truth. Labelling the entire datasets, which consist of multiple days of video, would be prohibitively expensive, which is why both prior work and our work opt for approximate ground truth. We argue that by using existing datasets and improving the label quality, we not only evaluated on existing benchmarks, but also went beyond what was done in prior work to minimize any negative effects from approximation. 

Finally, for this revision we have added a new dataset, BDD MOT, where MOT stands for multi-object tracking, a new video dataset released by the BDD creators that hand-labels object instances across frames. The main advantage of BDD MOT is that it has human labelled ground truth. The categories and videos it labels are different from the ones we used for BDD-1k, which is based on approximated ground truth from a model. We decided to keep both.  However, BDD MOT only labels 5 frames per second, consisting of about 1600 clips of short duration (around 30 seconds each, for around 150 labelled frames per clip), which is really not ideal for \sys, but nevertheless results in \autoref{fig:ratioplot} are similar to those BDD-1k.  \\
\reviewcomment{R3O3} {\em In the experimental results there is no plot showing the real gains in response time between \sys and competitors} \\
Thanks for pointing this out.  We now added \autoref{tab:times} to show the times for both proxy scanning (which is constant across queries, as it depends only on dataset size, so we only show it once per dataset) and \sys (we do not include \random in that table, because it would require 6 columns, and putting together \autoref{fig:ratioplot} with \autoref{tab:times}, can give a rough idea of absolute times for \random.). Note that in all cases our method takes less than the proxy scan time, including for reaching 90\% instances. \\
\reviewcomment{R3O4} {\em Future work is missing} We think there are three important areas for future work. One is designing better approaches to integrating NN models into the sampling decisions themselves so that we reap the benefits of more powerful models while avoiding an expensive scanning phase for every query. Another important area would be automating the chunking. A third area for work is standardizing a set of benchmarks, queries and ground truth for video analytics more broadly. As you highlight in your previous comment, we also feel the current state of this area is less than ideal, and building up a good set of benchmarks could help evaluating a lot of future work.  We have now added some of this future work in \autoref{sec:futurework}.
\noindent\rule{\columnwidth}{1pt}

\reviewcomment{R4O1} {\em In their analysis of the variance of the estimator, they made an assumption that objects are independent (top of page 5). Please give some analyses on the case objects are dependent or give more rationalities on this assumption.} This is an important question. The short answer is that we expect the events $X_i,X_j$ representing of two objects $i$ and $j$ appearing in a randomly sampled frame to not be always independent in reality. We found it helpful as a mental model to design the algorithm and reason about its variance, and in in our evaluation we found empirically that the end-to-end method does offer gains on real data. In the next two paragraphs we provide some analysis of the effects of dependencies and some further validation that shows that while real data has some degree dependencies which cause our variance estimate to be somewhat lower than the real variance, but still in the correct ballpark.

 Because our estimate is derived from a sum of random variables $X_i$, The variance of our estimate has the form $\text{Var}[\sum X_i]$. If we assume independence, we know the total variance is the sum of individual variances $\sum \text{E}[X_i^2] - \text{E}^2[X_i]$, which we will call $V_0$. If we do not assume independence of the $X_i$ then the definition of variance expands to  $\text{E}[(\sum X_i)^2] - (\text{E}[\sum X_i])^2$. This sum includes all the terms in $V_0$ above, plus pair-wise interaction terms $\text{Var}[\sum X_i] = V_0 + \sum_{i\neq j} \text{E}[X_iX_j] - \text{E}[X_i] \text{E}[X_j]$. The difference is the sum of over all pairs of instances of  $\text{E}[X_iX_j] - \text{E}[X_i] \text{E}[X_j]$.

	 From this expression we can deduce that if pairs of results tend to co-occur in frames, the variance $V_0$ from assuming independence will be an under-estimate.  If however results tend to not appear simultaneously $V_0$ will be an overestimate. Both effects are likely in real data. We intuitively expect cars to co-occur in frames for a street because of congestion.  At the same time, a picture frame has a finite visual range and can only fit a few objects so co-occurrence at the frame level is bounded in the worst case. 
	 
	According to the previous reasoning, dependencies imply our variance estimates could be off the mark in either direction. To clarify whether this is a significant problem we carried out the following experiment using real data for the objects in the BDD MOT (multi-object tracking) dataset, which as we have mentioned in R302 is a new dataset we have added to the benchmark and is labelled manually by the BDD creators. We took 1000 samples sequentially and using the available ground truth, tracked the {\em real} reward at each point in time $R(n)$, normally invisible to the algorithm,  as well the estimated belief distribution \autoref{eq:gamma} for that run, which you may recall incorporates both our expectation estimate of $R(n)$ as well as the variance estimate \autoref{eq:variance}. We ran each experiment 400 times and checked the fraction of runs for which the actual $R(n)$ lies within the 95\% confidence interval derived from the estimate \autoref{eq:gamma}. Because we can imagine the number of samples affecting our estimators accuracy, we evaluate confidence bounds after $n$=10, 100 and 1000 samples are taken. Results of this experiment appear in the table below. We expect the number on each cell to be more than or equal to .95, since it is a .95 confidence bound. The actual numbers can be somewhat lower, but most lie above .7, suggesting variance estimate being somewhat lower than the true variance, and in turn suggesting there is a higher result co-occurrence within frames (especially for car and pedestrian at .59 and .51 respectively) than we would expect from truly independent events.  regardless, the estimate is still remarkably accurate. This is due to car and pedestrian co-occurring often in the same frame. This inaccuracy is not very harmful because car and pedestrian are also very common categories. We have added a note about this experiment in \hlineref{line:dependence}. %\autoref{sec:algorithm}.

\begin{table}[!ht]
\center
\small
	\begin{tabular}{lrrr}
\toprule
 &  10   &  100  &  1000 \\
category      &       &       &       \\
\midrule
bicycle       &  0.47 &  0.87 &  0.76 \\
bus           &  0.82 &  0.90 &  0.88 \\
car           &  0.65 &  0.68 &  0.59 \\
motorcycle    &  0.97 &  0.84 &  0.88 \\
%other person  &  0.96 &  0.65 &  0.83 \\
%other vehicle &  0.99 &  0.88 &  0.85 \\
pedestrian    &  0.51 &  0.55 &  0.52 \\
rider         &  0.98 &  0.91 &  0.90 \\
trailer       &  0.98 &  0.59 &  0.81 \\
train         &  1.00 &  0.98 &  0.92 \\
truck         &  0.84 &  0.87 &  0.81 \\
\bottomrule
\end{tabular}
\label{tab:vartest}
\caption*{Fraction of runs for which the actual $R(n)$ lies within the 95\% confidence interval  (credible interval) using the mean and variance from \autoref{eq:gamma}. For each run we show this fraction for $n=10$, 100 and 1000. We expect the number on each cell to be .95. See R402. }
\end{table} 
\reviewcomment{R4O2} {\em Even though the authors point out the chunk number can affect the algorithm’s performance and show some relevant experimental results, they still regard that as an input parameter. Adding a function to automatically/adaptively select a good chunk number will make the algorithm more practical and stronger.}.   In \autoref{sec:chunks} and in particular in \autoref{fig:chunks}  we observe that beyond a point increasing the number of chunks can be detrimental, however \autoref{fig:chunks} also shows that 16, 128 and 1024 chunks all work well, and are orders of magnitude apart. Our takeaway is that even though the number of chunks can affect performance, there is a wide range that is acceptable.

As to how we arrive at these chunks for our benchmark in particular, as we mention also in R301, for datasets such as BDD we had little choice due to the nature of the BDD dataset (which is made up of many small clips), so we used one clip per chunk. For the other datasets we used 20 minute chunks. We did not see significant differences in \autoref{fig:ratioplot} whether we used 10, 20, or 40 minute chunks.

We agree an automated way of doing this would be an interesting contribution, and in fact experimented with several strategies for doing this.  The main goal of a chunking strategy would be to minimize the number of chunks while making each chunk highly homogeneous, and making different chunks very different from each other (from the point of view of probabilities). One natural way to try to do this would be to try to cluster chunks. The difficulty here is having metadata (e.g., locations) that could be used for clustering; there is some metadata for BDD but not the rest of the datasets, which makes this approach challenging. A different approach could be to start with coarse chunking and somehow split chunks dynamically as we run and gather statistics about chunks. We tried several variants of this second strategy but couldn't identify a method that worked well in general: for example, splitting chunks in half when skew is detected does not help if the skew occurs at a finer level, which is true for our datasets.  An alternative strategy of searching through a space of unequal splits runs into a problem with noise:  many potential splits will (falsely) look promising due to noise, even if we use error bounds that would be valid for a fixed partitioning. Thus,  because we found that a fixed chunking strategy worked well in practice and couldn't identify a general automated chunking strategy, we decided to leave chunking as an input parameter.
\\
\reviewcomment{R4O3} {\em In the Section V.A, Implementation paragraph, the authors mention proxy cannot reuse the scores for different object queries. Does this mean the proposed algorithm can do that? If yes, please give more explanations here.} The proposed algorithm does not reuse scores either, because it does not use that type of per-frame score in the first place, instead we score chunks of data. We changed the wording of \hlineref{line:score} to avoid suggesting this \\
%%%
\reviewcomment{R4 additional remarks} Thank you for taking time to enumerate these typos, we have corrected them.

\noindent\rule{\columnwidth}{1pt}
\reviewcomment{R5O1} {\em Experimental studies are weak.
Despite a large number of existing works, only naive sampling is used as baselines in experiments. It is thus unclear how this paper advances the techniques in the field in terms of effectiveness and efficiency. } In addition to uniform random sampling baselines, we evaluated \sys against BlazeIt, which is a state-of-the-art representative of proxy-based methods that have been shown to perform well on related problems~\cite{noscope,focus,blazeit}. These proxy-based methods use an importance sampling approach based on a score computed using NNs for each frame, and this frame scan dominates their cost, as we show now more clearly in \autoref{tab:times}. \sys avoids scanning the dataset, and this the main reason why it works better than sophisticated scoring. \\
\reviewcomment{R5O2} {\em The justification of excluding NN based solutions is not convincing. In the paper it is claimed that the NN based solutions are inferior since the models need to be trained for specific object types. However, in real-life applications, such as the camera datasets used in this paper, the number of data types is quite limited such as vehicles, pedestrians, traffic lights, etc and it is feasible to pre-train the model for different data sets. These solutions cannot be considered as completely orthogonal to this paper.} We did not intend in our paper to imply NN based solutions are inferior, and we do not exclude them in our evaluation. The current state of the art alternatives based on training a lighter weight NN for scoring frames before-hand suffer from a large fixed cost of scanning as a prerequisite to extracting a first sample, which we demonstrate through extensive experiments. The weakness in that approach is not meant to suggest NNs can only be used as black boxes, just that the particular way of using them embodied in current work has a limitation due to scanning which is absent in \sys. We have changed the wording in the introduction \hlineref{line:NNbased} to make this clearer.
We also should point out \sys in fact uses an NN based detector once our sampling algorithms have identified frames likely to contain an object of interest.   Our contribution is the sampling strategy and the way we integrate the expensive NN into this process: the NN itself is part of \sys, but we don't view using NNs for object detection as our contribution. 

We think a good area for future work is designing better approaches to integrating NN models into the sampling decisions themselves so that we get the benefits of these models but avoid an expensive scanning phase for every query. 
We added some ideas on how \sys could use better predictive scoring in \autoref{sec:futurework}.
\\

\fi %% reviewer comments

%\title[\sys]{\sys: Efficient Searches on Video Repositories through Adaptive Sampling} % SIGMOD
\title{\sys: Efficient Searches on Video Repositories through Adaptive Sampling % VLDB
\thanks{Research was sponsored by the United States Air Force Research Laboratory and the United States Air Force Artificial Intelligence Accelerator and was accomplished under Cooperative Agreement Number FA8750-19-2-1000. The views and conclusions contained in this document are those of the authors and should not be interpreted as representing the official policies, either expressed or implied, of the United States Air Force or the U.S. Government. The U.S. Government is authorized to reproduce and distribute reprints for Government purposes notwithstanding any copyright notation herein.
}
}
\if{0}
\numberofauthors{6}

\author{
\alignauthor
Oscar Moll\\
       \affaddr{MIT CSAIL}\\
       \email{orm@csail.mit.edu}
% 2nd. author
\alignauthor
Favyen Bastani\\
       \affaddr{MIT CSAIL}\\
       \email{favyen@csail.mit.edu}
\alignauthor
Sam Madden\\
       \affaddr{MIT CSAIL}\\
        \email{madden@csail.mit.edu}
\and     
\alignauthor
Mike Stonebraker\\
		\affaddr{MIT CSAIL}\\
        \email{stonebraker@csail.mit.edu}
\alignauthor
Vijay Gadepally\\
 \affaddr{MIT Lincoln Laboratory}\\
 \email{vijayg@ll.mit.edu}
\alignauthor
Tim Kraska\\
	\affaddr{MIT CSAIL}\\
    \email{kraska@mit.edu}
}
\date{30 August 2021}
\fi

\author{\IEEEauthorblockN{Oscar Moll}
\IEEEauthorblockA{\textit{MIT CSAIL}\\
Cambridge, MA, USA \\
orm@csail.mit.edu}
\and
\IEEEauthorblockN{Favyen Bastani}
\IEEEauthorblockA{\textit{MIT CSAIL}\\
Cambridge, MA, USA \\
favyen@csail.mit.edu}
\and
\IEEEauthorblockN{Sam Madden}
\IEEEauthorblockA{\textit{MIT CSAIL}\\
Cambridge, MA, USA \\
madden@csail.mit.edu}
\linebreakand %% see command def above 
\IEEEauthorblockN{Mike Stonebraker}
\IEEEauthorblockA{\textit{MIT CSAIL}\\
Cambridge, MA, USA \\
stonebraker@csail.mit.edu}
\and
\IEEEauthorblockN{Vijay Gadepally}
\IEEEauthorblockA{\textit{MIT Lincoln Laboratory}\\
Lexington, MA, USA \\
vijayg@ll.mit.edu}
\and
\IEEEauthorblockN{Tim Kraska}
\IEEEauthorblockA{\textit{MIT CSAIL}\\
Cambridge, MA, USA \\
kraska@mit.edu}
}

\maketitle
\setcounter{page}{1}
%%% icde format stuff below 
\thispagestyle{plain}
\pagestyle{plain}

% \linenumbers
% \linenumbers
% \linenumbersep 3pt\relax
% \nolinenumbers

\begin{abstract}
Capturing and processing video is increasingly common as cameras become cheaper to deploy.  At the same time, rich video-understanding methods have progressed greatly in the last decade. As a result, many organizations now have massive repositories of video data, with applications in mapping, navigation, autonomous driving, and other areas.

Because state-of-the-art object-detection methods are slow and expensive, our ability to process even simple ad-hoc object search queries (``find 100 traffic lights in dashcam video'') over this accumulated data lags far behind our ability to collect the data. Processing video at reduced sampling rates is a reasonable default strategy for these types of queries; however, the ideal sampling rate is both data and query dependent. We introduce \sys, a low cost framework for object search over un-indexed video that quickly processes search queries by adapting the amount and location of sampled frames to the particular data and query being processed.

\sys\ prioritizes the processing of frames in a video repository so that processing is focused in portions of video that most likely contain objects of interest. It approaches searching in a similar way to a multi-arm bandit problem where each arm corresponds to a portion of a video.  On large, real-world datasets, \sys\ reduces processing time by 1.9x on average and up to 6x over an efficient random sampling baseline. Moreover, we show \sys finds many  results long before sophisticated, state-of-the-art baselines based on proxy scores can begin producing their first results.

\end{abstract}

\begin{IEEEkeywords}
video data,sampling,object detection
\end{IEEEkeywords}

\section{Introduction}

Video cameras have become incredibly affordable over the last decade, and are ubiquitously deployed in static and mobile settings, such as smartphones, vehicles, surveillance cameras, and drones. Large video datasets are enabling a new generation of applications. For example, video data from vehicle dashboard-mounted cameras (dashcams) is used to train object detection and tracking models for autonomous driving systems~\cite{bdd}; to annotate map datasets such as OpenStreetMap with locations of traffic lights, stop signs, and other infrastructure~\cite{osmblog}; and to automate insurance claims processing by analyzing collision scene footage~\cite{nexar}.  

However, these applications must process large amounts of video to extract useful information. Consider the  task of finding examples of traffic lights---for example, to annotate a map---within a large collection of dashcam video collected from many vehicles. The most basic approach to evaluate this query is to run an object detector frame-by-frame over the dataset, and select frames where it detects one or more lights. Because state-of-the-art object detectors run at about 10 frames per second (fps) over 1080p video on modern high-end Graphics Processing Units (GPUs), scanning through a collection of 1000 hours of 30 fps video with a detector on a GPU would take 3000 GPU-hours. In the offline query case, which is the case we focus on in this paper, we can parallelize our scan over the video across many GPUs; however, as the rental price of a GPU is around \$0.50 per hour (for the cheapest $g4$ Amazon Web Services [AWS] instance in 2021) \cite{aws}, our cost for this one ad-hoc query would be \$1.5K, regardless of parallelism. Hence, this workload presents challenges in both time and monetary cost.  Note that accumulating 1000 hours of video represents just 10 cameras recording for less than a week.

A straightforward method for mitigating this issue is to reduce the number of sampled frames: for example,  run the detector only on one frame per second of video; since it is reasonable to assume all traffic lights are visible for more than one second. The savings are large compared to inspecting every frame: processing only one frame every second decreases costs by 30x for a video recorded at 30 fps.  
%In reality we find that for this sort of query random sampling is a better strategy, in terms of both accuracy and cost, than fixing a sampling interval such as 1 second ahead of time. 
Unfortunately, this strategy has limitations.  For example, the one frame out of every 30 that we look at may not show the light clearly, causing the detector to miss it completely; lights that remain visible in the video for long intervals, e.g., 30 seconds, would be seen multiple times unnecessarily, thereby wasting resources detecting those objects repeatedly; or worse still, we may miss other types of objects that remain visible for shorter intervals, and in general the optimal sampling rate is unknown and varies across datasets depending on factors such as whether the camera is moving or static and the angle and distance to the object. 

%As we show later on, a carefully tuned fixed frame sampling rate may balance accuracy and cost well for one specific dataset and application, but turn out to be too little or too much if we vary either.

In this paper, we introduce \sys, a video sampling technique designed to reduce the number of frames that need to be processed by an expensive object detector for object-search queries over large video repositories. \sys\ models this problem as one of deciding which frame from the dataset to look at next based on what it has seen in the past. Specifically,  \sys\ splits the dataset into temporal chunks (e.g., half-hour chunks), and maintains a per-chunk estimate of the probability of finding a new object in a frame randomly sampled from that chunk.
ExSample iteratively selects the chunk with the best estimate, samples a frame randomly from the chunk, and processes the frame through the object detector. ExSample updates the per-chunk estimates after each iteration so that the estimates become more accurate as more frames are processed.

\begin{highlightenv}
Recent state-of-the-art methods for optimizing queries over large video repositories have focused on training cheap, proxy models to approximate the outputs of expensive, deep neural networks~\cite{miris,focus,noscope,videomon,msr,blazeit}. However, these methods require training a proxy model for each new query and then running it on the dataset to choose which frames to look at with the more expensive reference detector. While helpful for some queries, for ad-hoc object queries these approaches impose a  scanning overhead required to score the frames, which can be hard to compensate for, even if the scores help find results faster after scanning.  This is especially problematic if the user only wants to find a few results (limit queries) \hlinelabel{line:NNbased}.
\end{highlightenv}

\sys addresses the limit query problem from a different perspective. Rather than select frames for applying the object detector based on proxy scores, \sys employs an adaptive sampling technique that eliminates the need for a proxy model and an upfront full scan.
A key challenge here is that, in order to generalize across datasets and object types, \sys must not make assumptions about how long objects remain visible to the camera and how frequently they appear. To address this challenge, \sys guides future processing using feedback from object detector outputs on previously processed frames.
Also, \sys gives higher weight to portions of video likely to contain objects that have not  already been seen. Thus, \sys avoids redundantly processing frames that contain objects that were previously seen in other frames.

We evaluate \sys\  a variety of search queries spanning different objects, different kinds of video, and different numbers of desired results.  We show savings in the number of frames processed ranging up to 6x on real data, with a geometric average of 1.9x across all settings, in comparison to an efficient random sampling baseline. In the worst case, \sys\ does not perform worse than \random\ sampling, something that is not always true of alternative approaches.
%Additionally, in comparison to BlazeIt~\cite{blazeit}, our method processes fewer frames to find the same number of distinct results in many cases, showing that ExSample provides higher precision than a proxy model in selecting frames with new objects; even in the remaining cases, ExSample still requires one to two-orders of magnitude less runtime (or, equivalently, dollar cost) because it does not require an upfront per-query scan of the entire dataset.
%, and so can avoid the pre-processing costs of proxy-based approaches.

In summary, our contributions are (1) \sys, an adaptive sampling method that facilitates ad-hoc object searches over video repositories, (2) a formal analysis of \sys's design, and (3) an empirical evaluation showing \sys is effective on real datasets  under real system constraints, and outperforms existing approaches for object search.

\section{Background}

In this section we review object detection, introduce distinct object queries, and discuss limitations of prior work.
%explain our main cost evaluation metric: frames processed by the object detector, and justify our main baseline: random sampling.

\subsection{Object Detection}

An object detector is a model that operates on still images, inputting an image and outputting a set of boxes within the image containing the objects of
interest. The amount of objects found will range from zero to arbitrarily many. Well known examples of object detectors include Yolo~\cite{yolov3} and Mask R-CNN~\cite{maskrcnn}.
In \autoref{fig:instances}, we show two example frames that were processed by an object detector, with yellow boxes indicating detection of traffic lights.

Object detectors with state-of-the-art accuracy in benchmarks such as COCO~\cite{coco} typically execute at around 10 frames per second on modern GPUs, though it is possible to achieve real-time rates by sacrificing accuracy~\cite{tfmodels,yolov3}.

In this paper we do not seek to improve on state-of-the-art object-detection approaches. Instead, we regard object detectors as a black box with a costly runtime; and aim to substantially reduce the number of video frames that need to be processed by the detector.

\subsection{Distinct Object Queries} \label{sec:objquery}

Also, we are interested in processing higher level queries on video enabled by the availability of object detectors. In particular, we are concerned with distinct object limit queries such as ``find 20 traffic lights in my dataset'' over large repositories of multiple videos from multiple cameras. For distinct object queries, each result should be a detection of a different object. For example, in \autoref{fig:instances}, we detected the same traffic light in two frames several seconds apart; although these detections are at different positions and are computed in different frames, in a distinct object query, these two detections yield only one distinct result.

%In natural video, any object of interest lingers within view over some length of time. For example, the frames in \autoref{fig:instances} contain the same traffic light a few seconds apart.  While either frame would be an acceptable result for our traffic light query, an application such as automatic street map annotation would not benefit more from having both frames. We are therefore interested in returning distinct object results, and we refer to this specific variation as a {\em distinct} object query. Similarly, for an application such as constructing a training set for a classifier or detector, richer sets of diverse examples are preferable to duplicates.  Note that an application could always find all detections of an object if desired by traversing the video backwards and forwards starting from a given result, the more difficult part is finding initial results in the first place.  

\begin{figure}[htb]
\includegraphics[width=\linewidth]{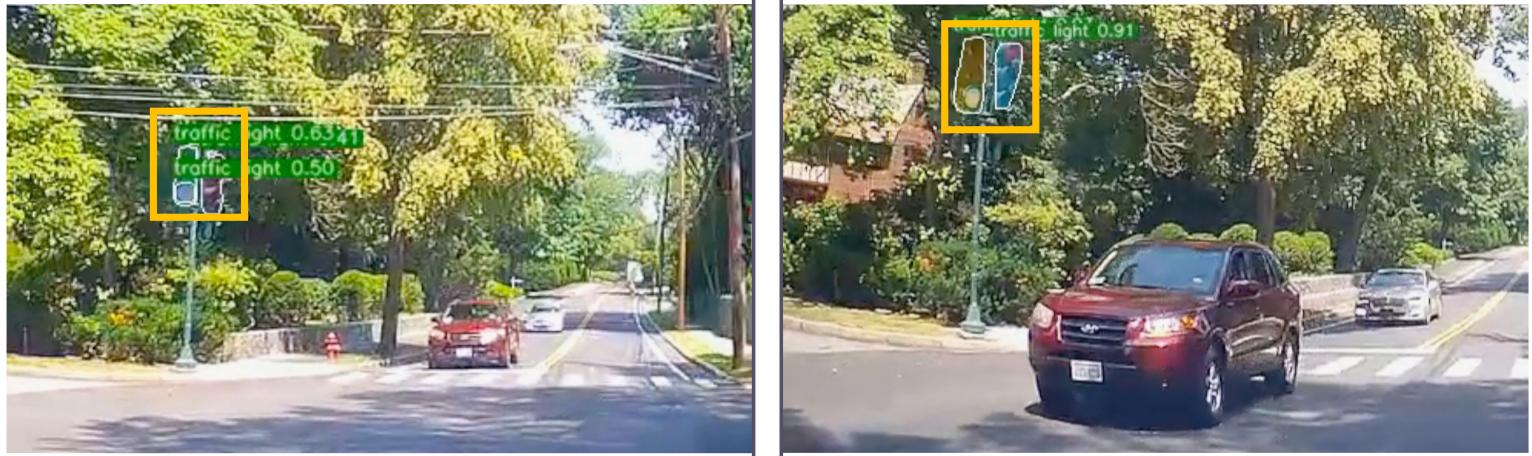}
\caption{Two video frames showing the same
  traffic light instance several seconds apart. A distinct object query is defined by having these two boxes only count as one result.}
%  The drawn boxes correspond to two rows in 
%%  the {\tt object\_detections} table of \autoref{lst:query}, with different {\tt frame\_id} but the same {\tt instance\_id}  }
\vspace{-.1in}
\label{fig:instances}
\end{figure}

Then, to define a query, users must specify not only the object type of interest and the number of desired results, but also a \emph{discriminator function} that determines whether a new detection corresponds to an object that was already seen earlier during processing. In this paper, we assume a fixed discriminator that applies an object-tracking algorithm to determine whether a detection is new. The discriminator applies a tracker similar to SORT~\cite{SORT} backwards and forwards through video for each detection of a new object to compute the position of that object in each frame where the object was visible; then, future detections are discarded if they match previously observed positions.

%For example, the discriminator function may apply an object tracking algorithm like Median Flow~\cite{medianflow} or SORT~\cite{SORT} that computes the position of an object over a sequence of frames; then, by tracking new instances backwards and forwards through video around the sampled frame in which they were first detected, we can determine whether future instances correspond to prior ones by comparing them against previously computed tracks.  
%We expand on this in our evaluation section.

The goal of this paper is to reduce the cost of processing such queries. Moreover, we want to do this on ad-hoc distinct object queries over large video repositories, where there are diverse videos in our dataset and where it is too costly to compute all object detections ahead of time for the type of objects we are looking for. This distinction affects multiple decisions in the paper, including not only the results we return but also the main design of \sys and how we measure result recall.

Following, we discuss two baselines and prior work for optimizing the processing of distinct object queries.

\smallskip
\noindent
\textbf{Naive execution.} A straightforward method is to process frames sequentially, applying the object detector on each frame of each video, using the discriminator function to discard repeated detections of the same object. Once enough distinct objects to satisfy the query's limit clause are collected, scanning can stop. A natural extension is to sample only one out of every $n$ frames.
%We can also reduce the sampling rate, so that we only process one in every $T$ frames, e.g., detect objects in one frame per second of video.
Sequential processing exhibits high variance in execution time due to the uneven distribution of objects in video. Moreover, if objects appear in the video for much longer than the sampling rate, we may repeatedly compute detections of the same object. Similarly, if objects appear for shorter than the sampling rate, we may completely miss some objects.

\smallskip
\noindent
\textbf{Random sampling.} A better strategy is to iteratively process frames uniformly sampled from the video repository (without replacement). This method reduces the query execution time over naive execution as it explores more areas of the data more quickly, whereas sequential execution can get stuck in a long segment of video with no objects. Additionally, early in query execution, randomly sampled frames are less likely to contain previously seen objects compared to frames sampled sequentially.

\smallskip
\noindent
\textbf{Proxy-based methods.} Methods that optimize video query execution by training cheap, proxy models to approximate the outputs of expensive object detectors have recently attracted much interest. In particular, BlazeIt~\cite{blazeit} proposes an adaption of proxy model techniques for processing distinct object queries with limit clauses. Rather than randomly sampling frames for processing through the object detector, BlazeIt processes frames in order of the score computed by the proxy model on the frame, beginning with the highest scoring frames. Since the proxy model is trained to produce higher scores on frames containing relevant object detections, this approach effectively ensures that frames processed early during execution contain relevant detections (but not necessarily new objects).

However, proxy-based methods have several critical shortcomings when used for processing object search queries. First, for queries seeking rare objects that appear infrequently in video, these methods require substantial pre-processing to collect annotations for training the proxy models, which can be as costly as solving the search problem in the first place; indeed, BlazeIt resorts to random sampling if the number of positive training labels is too small. Moreover, when processing limit queries, proxy-based methods require an upfront per-query dataset scan in order to compute proxy scores on every video frame in the dataset. As we will show in our evaluation, oftentimes, the cost of performing just this scan is already larger than simply processing a limit query under random sampling.

\section{\sys}
\label{sec:sys}
In this section we explain how our approach, \sys, optimizes query execution.  Due to the high compute demand of object detectors, runtime in ExSample is roughly proportional to the number of frames processed by the detector. Thus, we focus on minimizing the number of frames processed to find some number of distinct objects.
%\srm{just pick one and be consistent - i would favor processed.  then you don't need to say this.}  When discussing \sys\ we use the terms frames sampled and frames processed interchangeably because every sampled frame is also a processed frame.

To do so, \sys estimates which portions of a video are more likely to yield new results, and samples frames more often from those portions. Importantly, \sys prioritizes finding distinct objects, rather than purely maximizing the number of computed object detections.

At a high level, ExSample conceptually splits the input into chunks, and scores each chunk separately based on the frequency with which distinct objects were found in that chunk in the past. 
%Crucially, when \sys finds an instance it has seen before, it lowers the  the lower the score will be, regardless of past hits. 
This scoring system allows \sys to allocate resources to more promising chunks, while also allowing it to diversify where it looks next as more frames are sampled.  Our evaluation shows that this technique helps \sys outperform greedy, proxy-guided strategies, even when they employ duplicate avoidance heuristics (i.e., do not process frames that are close to previously processed frames).

To make it practical, \sys is composed of two core components: an estimate of future results per chunk, described in Section \ref{sec:derivation}, and a mechanism to translate these estimates into a sampling decision which accounts for estimate errors, introduced in Section \ref{sec:thompson}. In these two sections we focus on quantifying the types of error. Later, in Section \ref{sec:algorithm}, we detail the algorithm step by step.
	
\subsection{Scoring a Single Chunk}
\label{sec:derivation}
In this section we derive our estimate for the future value of sampling a chunk. 
To make an optimal decision of which chunk to sample next, assuming $n$ samples were taken, \sys estimates $R(n+1)$ for each chunk, which represents the number of {\em new} results we expect to find on the next sample. {\em New} means $R$ does not include objects already found in the previous $n$ samples, even if they also appeared in the  $(n+1)^{\text{th}}$ frame. Intuitively, the chunk with the largest $R(n+1)$ is a good location to pick a frame from next. 

Our main conceptual tool for reasoning about the quantity $R(n+1)$ as well as our estimates for it throughout this paper is to analyze each result instance in isolation and aggregate to obtain global estimates. More formally, let $N$ be the number of distinct objects in the data.  Each object $i$ out of those $N$ is visible for a different number of frames. When sampling frames at random from a chunk, each $i$ will have a different probability $p_i$ of being found proportional to its duration. For example, in video collected from a vehicle stopped at a red light, red lights will tend to have large $p_i$ lasting in the order of minutes, while green and yellow lights are likely to have much smaller $p_i$, perhaps in the order of a few seconds. We find in practice these quantities $p_i$ vary widely even for a single object class, from tens to thousands of frames. For any given run of samples, $R(n+1)$ is the sum over the $p_i$ of all as-yet unseen instances $R(n+1)=\Sigma_{i=1}^N p_i\cdot[i \notin \text{seen}(n)]$

We emphasize that $p_i$ and $N$ are used in order to reason about \sys, but they are {\em not known in advance} by \sys or by the user. Instead, a key contribution of \sys is to adapt to a dataset and a query specific set of $p_i$ during the sampling process. Furthermore, we clarify \sys also does not attempt to estimate individual $p_i$ or $N$, instead \sys estimates $R$ directly with the following estimate $\hat{R}$:
% \begin{definition}[Estimate $\hat{R}$]
\begin{equation}
\hat{R}(n+1) \mathrel{\mathop:}= \Ne(n)/n  \label{eq:estimate0}
\end{equation}
where $\Ne(n)$ is the number of distinct {\em results (objects)} seen exactly once so far, and $n$ is the number of {\em frames} sampled and inspected so far.
% \end{definition}

To implement this estimate \sys tracks how many times we have seen each distinct result. Each result seen only once contributes a count of one to $\Ne(n)$. Results seen more than once do not contribute anything to $\Ne(n)$, and neither do any unseen results (which would be impossible to account for directly). Unlike the $p_i$ or $N$, which are unknown to the user and to \sys, $\Ne(n)$ is a quantity \sys can observe.  

We now justify the definition of $\hat{R}$ in \autoref{eq:estimate0} by bounding the expected error $\E[\hat{R} - R]$ in terms of data dependent quantities (we let $\mu$ and $\sigma$ stand for mean and standard deviation)

\begin{theorem}[Bias of $\hat{R}$]
\begin{equation} \label{eq:bias}
0 \leq \frac{\E\left[\hat{R} - R\right]}{\hat{R}} \leq 
\begin{cases}
	\max{p_i}\\
	\sqrt{N}(\mu_p + \sigma_p)
\end{cases}
\end{equation}
\end{theorem}

Intuitively, \autoref{eq:bias} states $\hat{R}$ overestimates $R$ in expectation, that the size of the overestimate is guaranteed to be less than the largest probability, which is likely small, but even if some $p_i$ outliers made $\max p_i$ large,  if the standard deviation $\sigma_p$ is small we can still bound the bias.  \autoref{eq:bias} does suggest (but does not imply) that skew  (measured by $\sigma_p$) could affect the accuracy of the estimate, we will address that experimentally. A large $N$ or a large $\mu_p$ may seem problematic for \autoref{eq:bias}, but we note that having large number of results $N$ or long average duration $\mu_p$ implies many results are found after only a few samples, so the end goal of searching through the data is easier in the first place and guaranteeing accurate estimates is less important.

\begin{proof} 
The main proof idea is accounting for each object $i$'s individual expected contribution to both $R(n+1)$ and to $\Ne(n)$. We then recover both $\Ne(n)$ and $R(n+1)$ by adding each term, justified by linearity of expectation.
The individual inequalities of \autoref{eq:bias} follow from analyzing this sum. The event that on a given sequence of samples object $i$ is seen on the $(n+1)^{\text{th}}$ sample after being missed the on the first $n$ samples occurs with probability $p_i(1 - p_i)^{n}$. We denote this quantity $\pi_i(n+1)$.  On the other hand, the event that after $n$ samples object $i$ has appeared, exactly one occurs with probability $np_i(1-p_i)^{n-1} = n\pi_i(n)$.  Therefore, $E\left[R(n+1)\right] = \sum \pi(n+1)$, while $E\left[\Ne(n)\right] = n \sum \pi(n)$. These sums are taken over the N object indices $i$, which we omit for convenience. Therefore, $E\left[\Ne(n)/n - R(n+1)\right] = \sum \pi(n) - \pi(n+1)$. Because by definition $\pi(n+1) = (1-p)\pi(n)$, the right-hand side simplifies to $\sum p\pi_(n)$. Intuitively, term by term this error is small compared to the terms in our estimate: $\sum \pi (n)$, especially when the $p_i$ are small. The expected error is positive, showing the left side of \autoref{eq:bias}. The top right inequality is easily seen by replacing the individual $p$ with $\max {p}$ and factoring it out. 

The remaining inequality  can be derived by applying Cauchy-Schwartz to the sum.
\iftechreport
\begin{align*} 
	\Sigma p_i\pi_i(n) &\leq \sqrt{(\Sigma p_i^2)(\Sigma \pi_i^2)} & \text{(Cauchy-Schwarz)} \nonumber\\
	&\leq \sqrt{(\Sigma p_i^2) \left(\Sigma \pi_i\right)^2} & \pi_i \text{ are positive } \nonumber\\
	&= \sqrt{\left(\Sigma p_i^2\right)} \Sigma \pi_i\nonumber\\
	&= \sqrt{\left(\Sigma p_i^2\right)} \E[\Ne]/n \nonumber\\
	&= \sqrt{N} \sqrt{\left(\Sigma p_i^2/N\right)} \E[\Ne]/n \nonumber\\
	\intertext{The term $\Sigma p_i^2/N$ is the second moment of the $p_i$, and can be rewritten as $\sigma^2_p + \mu^2_p$, where $\mu_p$ and $\sigma_p$ are the mean and standard deviation of the underlying result durations:}
		&= \sqrt{N}\sqrt{(\sigma^2_p + \mu^2_p)}\E[\Ne]/n \nonumber\\
			&\leq \sqrt{N}(\sigma_p + \mu_p)\E[\Ne]/n  
			\label{eq:proof_sigp}
\end{align*}
\else
The full details appear in an extended version \cite{techreport}. For this paper, the main importance of this last inequality involving $\sigma_p$ is to consider the relevance of skew over the $p_i$ in practice.
\fi
\end{proof}

%\subsubsection{Edge cases}
%Instances spanning multiple chunks: nothing prevents an instance from appearing across multiple different chunks. For example a single instance may occur right in the last frames of one chunk and then on to the next one, or oven occur across more than two of them. In our implementation we only credit the first chunk to have sampled the instance, and we remove the credit if any second chunk happens to find the same instance. The correctness of the policy is more clear if we consider a scenario with three chunks where two hold the same instances with the same probabilities, and a third has a completely disjoint set of results with the same probabilities. In this scenario, this policy will split one half the samples to the third chunk, and the other half will be allocated to the identical chunks, which is the optimal thing to do in this scenario (could be considered one chunk). 

\subsection{Picking the best chunk under uncertainty} \label{sec:thompson}
In this section we explain how we extend the earlier estimate with an estimate of its uncertainty so we can use it to make decisions about which chunk to use even when there is some error in our estimates.

If we knew  the real $R_j$ for every chunk $j$, then the optimal algorithm would simply sample from the chunk with the largest $R_j$ value. However, if we  use the raw point estimate $\hat{R}_j$ in place of $R_j$ \sys could get stuck sampling chunks with an early lucky result and ignore better chunks with unlucky early results. \footnote{Beside the uncertainty problem, when sampling from multiple chunks there is also a second, less important concern of how to handle instances spanning multiple chunks.  In~\cite{techreport} we show \autoref{eq:estimate0} can be adjusted to handle this problem.}

Now we explain how \sys  handles the problem that an observed $\Ne(n)$ fluctuates randomly due to randomness in our sampling. This is especially true early in the sampling process, where only a few samples have been collected but we need to make a sampling decision. Because the quality of the estimates themselves is tied to the number of samples taken, and we do not want to stop sampling a chunk due to a small amount of bad luck early on, it is important we estimate how noisy our estimate is. The usual way to do this is by estimating the variance of our estimator: $\Var[\Ne(n)/n]$, which we find is small compared to $\hat{R}$. 

\begin{theorem}[Variance]
\begin{equation} \label{eq:variance}
 \Var[\hat{R}(n+1)] \leq \E[\hat{R}(n+1)]/n
\end{equation}
\end{theorem}

\begin{proof}
	\hlinelabel{line:independence} Similar to our estimate of bias, we estimate the variance $\Ne(n)$ assuming independence of the different instances and adding their individual variances.
	 We can express $\Ne(n)$ a sum of binary indicator variables $X_i$, which are 1 if instance $i$ has shown up exactly once. $X_i = 1$ with probability $n\pi_i(n)$. Therefore, $\Ne(n) = \sum_i X_i$ and because of our independence assumption $\text{Var}[\Ne(n)] = \sum_i \text{Var}[X_i]$. Because $X_i$ is a Bernoulli random variable, its variance is $n\pi_i(n)(1- n\pi_i(n))$ which is bounded by $n\pi_i(n)$ itself.  Therefore, $\text{Var}[\Ne(n)] \leq \sum n \pi (n)$. This latter sum we know is $\text{E}[\Ne(n)]$. Therefore, $\Var[\Ne(n)/n] \leq \text{E}[\Ne(n)]/n^2 = \E[\hat{R}(n+1)]/n $.
\end{proof}

In fact, we can go further and fully characterize the distribution of values $\Ne(n)$ takes as follows.

\begin{theorem}[Sampling distribution of $\Ne(n)$]
\label{th:poisson}
Assuming $p_i$ are small or $n$ is large, and assuming independent occurrence of instances,  $\Ne(n)$ follows a Poisson distribution with parameter $\lambda = \Sigma \pi_i(n)$.
\end{theorem}

\begin{proof}
The proof idea is to treat each instance separately as we did for variance and bias, but this time focus on the moment generating functions. The independence assumption lets us recover the moment generating function of $\Ne(n)$ from multiplying the individual functions.

\iftechreport
	We want to show $\Ne(n)$'s moment generating function (MGF) matches that of a Poisson distribution with $\lambda$: $ M(t) = \exp\left(\lambda [e^{t}-1]\right)$.
	
	As in the proof of \autoref{eq:variance}, we think of $\Ne(n)$ as a sum of independent binary random variables $X_i$, one per instance. Each of these variables has a moment generating function $M_{X_i}(t) = 1+ \pi_i(e^{t} - 1)$.
		Because $1+x \approx \exp(x)$ for small $x$, and $\pi_i(e^{t} - 1)$ will be small, then $1+ \pi_i(e^{t} - 1) \approx \exp(\pi_i(e^{t} - 1))$.  Note $\pi_i(e^{t} - 1)$ is always eventually small for some $n$ because $\pi_i(n+1) = p_i (1-p_i)^{n} \leq p_i e^{-np_i} \leq 1/en$.
		
		Because the MGF of a sum of independent random variables  is the product of the terms' MGFs, we arrive at: \\
		 $M_{\Ne(n)} = \prod_i M_{X_i}(t) = \exp\left( \left[\sum_i \pi_i\right] \left[e^{t} - 1\right]\right)$
\else
 Refer to \cite{techreport} for full details.
\fi
\end{proof}

\iftechreport
\subsubsection{Instances spanning multiple chunks}
A corner case when deciding which chunk to sample from is how to deal with instances spanning multiple chunks: how do we update the counts?

If instances span multiple chunks, for example a traffic light that spans across the boundaries of two neighboring chunks, \autoref{eq:estimate0} is still accurate with the caveat that $\Ne_j(n_j)$ is interpreted as the number of instances seen exactly once {\em globally} and which were found in chunk $j$. The same object found once in two chunks $j$ and $k$ does not contribute to either $\Ne_j$ or to $\Ne_k$, even though each chunk has seen it only once. The derivation of this rule is similar to that in the previous section. In practice, if only a few rare instances span multiple chunks then results are almost the same and this adjustment does not need to be implemented.

At runtime, the numerator $\Ne_{j}$  only increases in value only the first time we find a new result globally, decreases back as soon as we find it again either in the same chunk or elsewhere, and finding it a third time does not change it. Meanwhile $n_j$ increases upon sampling a frame from that chunk. This natural relative increase and decrease of $\Ne_j$ with respect to each other allows \sys\ to seamlessly shift where it allocates samples over time.

Here we prove \autoref{eq:estimate0} is also valid when different chunks may share instances. Assume we have sampled $n_1$ frames from chunk 1, $n_2$ from chunk 2, etc. Assume instance $i$ can appear in multiple chunks: with probability $p_{i1}$ of being seen after sampling chunk 1, $p_{i2}$ of being seen after sampling chunk 2, and so on.  We will assume we are working with chunk 1, without loss of generality. The expected number of new instances if we sample once more from chunk $1$ is:

\begin{equation*}
	R_1(n+1) = \sum_{i=1}^{N}\left[ p_{i1}(1-p_{i1})^{n_1}
 \prod_{j=2}^{M} (1-p_{ij})^{n_j}\right]
\end{equation*}

Similarly, the expected number of instances seen exactly once in chunk $1$, and in no other chunk up to this point is

\begin{equation*}
	\Ne_1 = \sum_{i=1}^{N} \left[ n_1 p_{i1}(1-p_{i1})^{n_1 - 1}\prod_{j=2}^{M} (1-p_{ij})^{n_j} \right]
\end{equation*}

In both equations, the expression $\prod_{j=2}^{M} (1-p_{ij})^{n_j}$ factors in the need for instance $i$ not to occur while sampling chunks 2 to $M$. We abbreviate this factor as $q_i$. When instances only show up in one chunk, $q_i = 1$, and everything is the same as in \autoref{eq:estimate0}.

The expected error is:
\begin{equation*}
	\Ne_1(n_1)/n_1 - R_1(n_1+1) = \sum_{i=1}^{N} \left[p_{i1}^2(1-p_{i1})^{n_1 - 1}q_i\right]
\end{equation*}

which again is term-by-term smaller than $\Ne_1(n_1)/n_1$ by a factor of $p_i$
\fi %techreport

\subsection{Thompson Sampling}
Now we use the previous results to design a sampling strategy.  The goal is to  pick among chunks $j$, balancing both the desire for a large $\hat{R_j}(n_j)$ and the unreliability of the point estimate when $n_j$ is small. Thompson sampling~\cite{thompson_sampling} is one way to automate this process. Thompson sampling works by modeling unknown quantities such as $R_j$, not just with a point estimate such as $\hat{R_j}$, but with a distribution over its possible values informed by the uncertainty in the estimate. Then, instead of using the point estimate $\hat{R}$ directly to guide decisions, we use a sample from $\hat{R}'s$ distribution. In our implementation, we choose to model the uncertainty of our estimate $\hat{R}_j(n+1)$ as following a Gamma distribution. A Gamma distribution is shaped much like the Normal, except it is restricted to non-negative values. When its mean is a high positive number compared to its standard deviation it appears much like a Normal. When its mean is near 0 then its shape changes to have single mode at 0 (we show some different shapes in \autoref{fig:estimates}, in solid red). Among other uses, the Gamma distribution is a way to model the uncertainty of the hidden, positive parameter $\lambda$ of a Poisson distribution which we only get to  interact with through sampling. The Gamma distribution is fully described by parameters $\alpha$ and $\beta$, where both are positive real numbers. We use $\alpha=\Ne_{j}$ and $\beta=n_j$ because the mean for such a distribution is $\alpha/\beta$, or $\Ne_{j}/n_j$ in our case, which is by construction consistent with \autoref{eq:estimate0}; and the variance is $\alpha/\beta^2$, or $\Ne_{j}/n_j^2$ in our case, which is by construction consistent with the bound \autoref{eq:variance}. Finally, the Gamma distribution is not defined when $\alpha$ or $\beta$ are 0, so we need both a way to deal with the scenario where $\Ne(n) = 0$ which  happens at the start, when objects are rare, or when only a few objects are left.  We do this by adding a small quantity $\alpha_0$ and $\beta_0$ to both terms, obtaining the distribution below.

\begin{align}
\label{eq:gamma}
R_j(n_j+1) \sim \Gamma(\alpha =\Ne_j + \alpha_0, \beta=n_j + \beta_0)
\end{align}

We used $\alpha_0 = .1$ and $\beta_0 = 1$ in practice, though we did not observe a strong dependence on this value choice. 
We also experimented with alternatives to Thompson sampling, specifically Bayes-UCB~\cite{bayesucb}, which uses upper distribution quantiles based also on \autoref{eq:gamma} instead of samples to make decisions,  but again we did not observe different results.

\subsection{Empirical Validation}
\label{sec:emp-simulation}
In this section we provide an empirical validation of \autoref{eq:estimate0}, and \autoref{eq:gamma}. The question we are interested in is: given an observed $\Ne$ and $n$, what is the true $R(n+1)$, and how does it compare to the belief distribution $\Gamma(\Ne,n)$ from \autoref{eq:gamma}.

 We ran a series of simulation experiments. We first generate 1000 $p_1,...p_{1000}$ at random to represent 1000 results with different durations. To ensure there is skew in the $p$ we use a lognormal distribution to generate the $p_i$. To illustrate the skew in the values, the smallest $p_i$ is $3\times10^{-6}$, while the $\max p_i = .15$.  The parameters $\mu_p$ and $\sigma_p$ are $3\times10^{-3}$ and $8\times10^{-3}$ respectively. For a dataset with $1$ million frames (about 10 hours of video), these durations correspond to objects spanning from $1/10$ of a second up to about $1.5$ hours, more skew than what normally occurs.

Then, we model the sampling of a random frame as follows: each instance is in the frame independently at random with probability $p_i$. To decide which of the instances will show up in our frame, we simulate tossing 1000 coins independently, each with their own $p_i$, and the positive draws give us the subset of instances visible in that frame. We then proceed drawing these samples sequentially, tracking the number of frames we have sampled $n$ and how many instances we have seen exactly once, $\Ne$. We also record $\text{E}\left[R(n+1)\right]$: the expected number of new instances we can expect in a new frame sampled, which is possible because we can compute it directly as $\Sigma_{i}^{N} [i  \notin	\text{ seen }(n)] \cdot p_i$ because in the simulation we know the remaining unknown instances and know their hidden probabilities $p_i$, so we compute $\text{E}\left[R(n+1)\right]$. We sample frames up to $n=180000$, and repeat the experiment 10K times, obtaining hundreds of millions of tuples of the form $(n, \Ne, R(n+1))$ for our fixed set of $p_i$.   Using this data, we can answer our original question: given an $(\Ne, n)$ pair, what is the histogram of the actual $R(n+1)$ across runs and how does it compare to \autoref{eq:gamma}. We show these histograms for six pairs of $n$ and $\Ne$ in \autoref{fig:estimates}.
	 
	\begin{figure}[!ht]
\includegraphics[width=3.2in, trim={.cm .cm .cm .cm},clip]{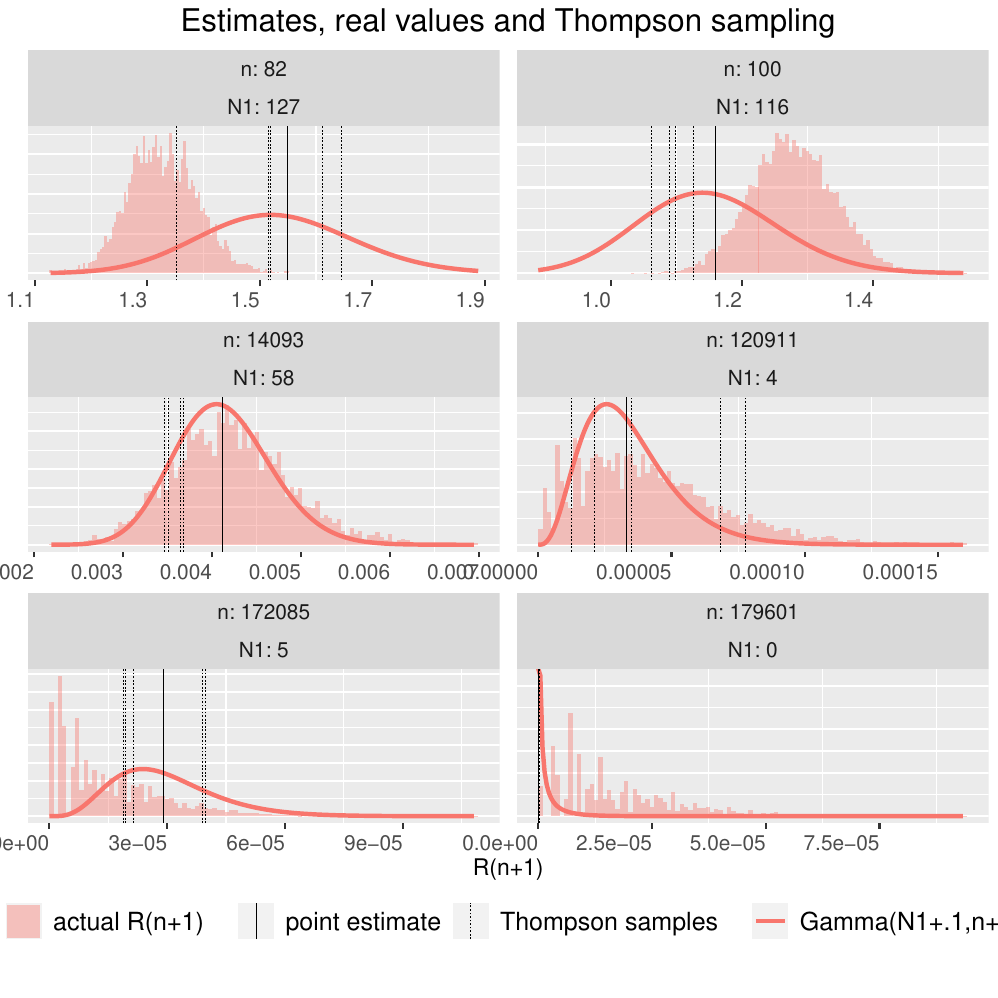}
\caption{Comparing our Gamma heuristic of \autoref{eq:gamma} with a histogram of the true values $R(n+1)$ from a simulation with heavily skewed $p_i$. The details of the simulation are discussed in Section \ref{sec:emp-simulation}. 
%We picked 10 $(\Ne,n)$ pairs from the data to include multiple important edge scenarios: where $n$ is less than $100$ as well as when $n$ is very large (in this case up to 20\% of the total frames). We also show $\Ne$ close to 0 due to bad luck in the last subplot. Note we are using the noisy observed $\Ne$ and not the idealized $\E[\Ne]$ which would be a lot more accurate, but is not directly observable. 
The histograms show the range of values seen for $R(n+1)$ when we have the observed $\Ne$ and $n$. The point estimate $\Ne/n$ of (\autoref{eq:estimate0}) is shown as a solid vertical line. The belief distribution density of \autoref{eq:gamma} is plotted as a thicker orange line, and five samples drawn from that distribution are shown as  dashed vertical lines. 
}

\label{fig:estimates}
\end{figure}	

	\autoref{fig:estimates} shows a mix of three important scenarios. The first two subplots with $n \leq 100$, representative of the early stages of sampling. Here we see the $\Gamma$ model has substantially more variance than the underlying true distribution of $R(n+1)$. This is intuitively expected: when $n=0$ the value of $R(1)$ is $\sum p$, which our estimator does not know. As $n$ grows to mid range values (next two plots), we see that the curve fits the histograms very well, and also that the curve keeps shifting left to lower and lower orders of magnitude on the x axis. Here we see that the one-sided nature of the Gamma distribution fits the data better than a bell shaped curve.
	The final two subplots show scenarios where $n$ has grown large and $\Ne$ potentially very small, including a case where $\Ne = 0$. In that last subplot, we see the effect of having the extra $\alpha_0$ in \autoref{eq:gamma}, which means Thompson sampling will continue producing non-zero values at random and we will eventually correct our estimate when we find a new instance. The bias error is not large enough despite the skew in $p$. 
	
	\begin{highlightenv}
	Finally, in reality it is possible that the different events  are not independent of each other, which is an assumption going into the variance estimate \autoref{eq:variance}.	\hlinelabel{line:dependence}  We tested the variance estimate directly on the objects of the Berkeley Deep Drive (BDD) multi-object tracking (MOT) dataset, and found that the 95\% confidence bound derived from the \autoref{eq:variance} includes the actual expected reward about 80\% of the time (with some variation across classes). In general, the trend was that our variance estimate is a slight underestimate, suggesting some dependency (co-occurring of events) causes the extra variance observed in practice.

	\end{highlightenv}

\subsection{Algorithm} \label{sec:algorithm}
Having derived and demonstrated how to estimate $R$ for different chunks, we now explain \sys step by step in \autoref{alg:pseudocode}, without having any prior information about $N$ or $\vec{p}$.

\begin{algorithm}
\SetKwInOut{Input}{input}\SetKwInOut{Output}{output}
\SetKwFunction{len}{len}
\DontPrintSemicolon
\Input{video, chunks, detector, discrim, result\_limit}
\Output{ans}
 $\text{ans} \gets$ [],
%  \BlankLine
%  \tcp{ arrays for stats of each chunk }
 $\Ne \gets$ [0,0,$\dots$,0],
 $n \gets $ [0,0,$\dots$,0]\;
 
 \While{\len{$\text{ans}$} $<\text{result\_limit}$}{
   \tcp{1) choice of chunk and frame}
    \For{$j\gets1$ \KwTo $M$}{
    	$\text{R}_j \gets \Gamma(\Ne[j] + \alpha_0,n[j] + \beta_0).\text{sample}()$\label{line:ts}\BlankLine}
    	
   $j^* \gets \argmax_j {R}_j$\;\label{line:best}
   $\text{frame\_id} \gets \text{chunks}[j^*].\text{sample}()$\;\label{line:chunk}
   \BlankLine
   \tcp{2) io,decode,detect,match}
    $\text{rgb\_frame} \gets \text{video}.\text{read\_and\_decode}(\text{frame\_id})$\;\label{line:io}
    $\text{dets} \gets \text{detector}(\text{rgb\_frame})$\;\label{line:det}
   \tcp{ $ \text{d}_0$ are the unmatched dets}
   \tcp{ $ \text{d}_1$ are dets with only one match}
   $\text{d}_0,\text{d}_1 \gets \text{discrim}.\text{get\_matches}(\text{frame\_id},\text{dets})$\;\label{line:match}
   \BlankLine
   \tcp{3) update state}
   $ \Ne[j^*] \gets \Ne[j^*] + $\len{$\text{d}_0$} $-$ \len{$\text{d}_1$}\;\label{line:updateN}
   $n[j^*] \gets n[j^*] + 1$\;\label{line:updaten}
   $\text{discrim}.\text{add}(\text{frame\_id},\text{dets})$\;\label{line:add}
   ans.add($\text{d}_0$)

}

\caption{\sys}
\label{alg:pseudocode}
\end{algorithm}

The inputs to the algorithm are:

\begin{itemize} 
\renewcommand\labelitemi{-}
\item \texttt{video}: The video data, either a single video or a collection of files. 
\item \texttt{chunks}: How we have partitioned the frames in the video.  There are $M$ chunks total. 
\item \texttt{object detector}:  Processes video frames and returns object detections (boxes) relevant to the query.
\item \texttt{discrim}: decides whether a detection is new or matches previous detections.
\item \texttt{result\_limit}: an indication of when to stop. 
\end{itemize}

The loop consists of three parts: picking a frame, processing the frame, and updating the sampler state. First, to pick a frame ExSample decides which frame to process next. It applies the Thompson sampling step in \autoref{line:ts}, where we draw a separate sample $R_j$ from the belief distribution \autoref{eq:gamma} for each of the chunks, which is then used in \autoref{line:best} to pick the highest scoring chunk. As in the rest of the paper, $j$ indexes over chunks.  During the first execution of the \texttt{while} loop all the belief distributions are identical, but Thompson sampling effectively breaks ties at random. Once decided on a best chunk index $j^*$, we sample a frame index at random from the chunk in \autoref{line:chunk}.

Second, ExSample reads and decodes the chosen frame, and applies the object detector on it (\autoref{line:det}). Then, we pass the computed detections on to a discriminator, which compares the detections with objects returned earlier during processing in other frames, and decides whether each detection corresponds to a new, distinct object. The discriminator returns two subsets: $d_0$, the detections that did not match with any previous results (new objects), and $d_1$, the detections that matched exactly once with any previous detection. We only use the size of those sets to update our statistics in the next stage. This frame processing stage of ExSample is the main bottleneck, and is dominated first by the detector call in \autoref{line:det}, and second by the random read and decode in \autoref{line:io}.
% 	In comparison, the overhead of the first part is negligible and fully parallelizable. It only grows with the number of chunks.

Third, we update the state of our algorithm, updating $\Ne$ and $n$ for the chunk we sampled. We store detections in the discriminator and append the new detections to the result set (\texttt{ans}). The amount of tracked state is proportional to the number of chunks, and to the number of results we detected so far.

\if{0}
\section{Notation Summary}
\label{sec:notation}
\noindent\begin{tabular}%
	{>{\raggedleft\arraybackslash}p{.5in}|%
	p{2.6in}%
}
$n$ & Number of frames sampled so far. Frames sampled and frames processed means the same thing.\\
 $N$ & Number of distinct results in the data. We treat the terms ``result'' and ``instance'' as synonyms.\\
 $i$ & Index variable over results. $i \in [1,N]$. \\
 $\Ne(n)$ & Number of results seen exactly once up until the $n^{th}$ sampled frame. We omit $n$ when is clear from context \\
 $N(n)$ & Number of unique results seen up to the $n^{th}$ sample.\\
 $\text{seen}(n)$ & set of $i$ seen after
 $n$ frames have been processed.\\ 
 $p_i$ & Probability of seeing result $i$ in a randomly drawn frame. It is proportional to duration in video. We treat duration and probability as synonyms.\\
  $R(n)$ & Number of new results we expect to find in the next sampled frame $R(n) = 
 \sum [i \notin \text{seen}(n-1)]\cdot p_i$  \\
 $\mu_p$ & mean duration $\Sigma p_i/N$ \\
 $\sigma_p$ & std. deviation of durations $p_i$ \\
 $\pi_i(n)$ & $p_i(1-p_i)^{n-1}$: the chance result $i$ appears first at the $n^{th}$ sampled frame. We may leave the $n$ implicit.\\
 $M$ & Number of chunks \\
 $j$ & Index variable over chunks. $j \in [1,M]$ \\
\end{tabular}
\fi
		
\subsection{Other Optimizations}
The implementation supplements \autoref{alg:pseudocode} with two  optimizations that integrate easily:

\vspace{.1in}
\noindent{\bf Batched sampling.}
\autoref{alg:pseudocode} processes one frame at a time, but on modern GPUs inference throughput is faster when performed on batches of images.  The code for a batched version is similar to that in \autoref{alg:pseudocode}, but on \autoref{line:ts} we draw $B$ samples per chunk $j$ instead of one sample from each belief distribution. The $\argmax$ code in \autoref{alg:pseudocode} then produces $B$ batch indices. The distribution over the indices  depends on Thompson sampling. The state update can also be done in batch form: all the updates to $\Ne_j$ and $n_j$ are commutative because they are additive.

\vspace{.1in}
\noindent{\bf Avoiding near duplicates within a single chunk.}
While \random\ sampling is both a reasonable baseline and a good method for sampling frames within selected chunks, \random\ allows samples to happen very close to each other in quick succession: for example in a 1000 hour video, random sampling is likely to start sampling frames within the same one hour block after having sampled only about 30 different hours, instead of after having sampled most of the hours once. For this reason, we introduce a variation of random sampling, which we call \lsb, to deliberately avoid sampling temporally near previous samples when possible: by sampling one random frame out of every hour, then sampling one frame out of every not-yet sampled half hour at random, and so on, until eventually sampling the full dataset. We evaluate the separate effect of this change in our evaluation. We also use \lsb\ to sample within a chunk in ExSample, by modifying the internal implementation of the \texttt{chunk.sample()} method in \autoref{line:chunk}.

\section{Determinants of \sys limits }
\label{sec:analysis}
We have explained how \sys makes decisions on where to sample and how it adapts its sampling based on statistical estimates. In this section, we explain how much better than random \sys can be, and when it cannot be much better than random.  We first show a different (non practical) sampling method that upper bounds how well \sys can do. We show though simulation that \sys matches this upper bound in practice (though only after enough sampling has happened).  Second, we identify two important factors that affect how much better than random sampling \sys can be: the first factor is skew in how results are spread within the dataset, which is inherent to the data; the second is the number of chunks the data is split into, which the user chooses ahead of time. Using the same simulation we show how these factors affect \sys performance.
%% I want the formula to be explained here.
%% then some intuition about it. 
\subsection{Optimal chunk weights}
\label{sec:upperbound}
Here we derive an alternative way of sampling chunks that upper bounds \sys. Though it is not practical, it is helpful to understand when \sys would not be helpful and also whether \sys is choosing optimally rather than simply better than random.

\sys can be thought of as implementing a type of weighted sampling (though the weights are not explicitly computed). At any point during a run, \sys allocates $n$ samples across $M$ different chunks, de-facto assigning a weight of $n_j/n$ to each chunk.  A natural question is how this assignment fares against an optimal weight assignment chosen ahead of time (assuming a fixed target $n$). In this section we derive the optimal weight assignment (as a function of $n$)  which is also a good conceptual benchmark for \sys. This benchmark is not applicable in real scenarios, but helps to understand \sys and its limits. 
The number of results discovered by random sampling after $n$ samples follows the curve $N(n) = \sum 1 - (1-p_i)^n$. If we first split the data into $M$ chunks, conceptually we can think of instance $i$ having an M-dimensional vector $\vec{p} = (p_{ij})$, coordinate $j$ in this vector is the conditional probability of seeing instance $i$ when sampling from chunk $j$. The chance of sampling from a chunk is also an M dimensional vector $\vec{w} = (w_j)$, the total chance of sampling instance $i$ is the dot product $\vec{p_iw}$. The expected number of instances found after $n$ samples in that scenario is $\sum 1- (1-\vec{p_iw})^n$.  The optimal offline allocation of samples to chunks is therefore

\begin{equation}
\begin{split}
	\argmax_{\vec{w}}  \sum 1 - (1-\vec{p_iw})^n  \label{eq:optimal}
\end{split}
\end{equation}

where $\vec{w}$ is constrained to be a valid weight vector (entries are non negatives and add up to 1). In the unrealistic scenario where $\vec{p_i}$ for all chunks were known then we can solve for $w$ using a package such as \cite{cvxpy}. We later show empirically in \autoref{fig:chunks} that \sys matches this benchmark as more samples are taken.   

Note uniform random sampling corresponds to equal weights and is optimal when all chunks share similar probabilities. However, when some chunks have more results than others or when the $p_i$ change across chunks we expect \sys will discover better sample allocations.
	
\subsection{Instance skew across chunks}
\label{sec:simulation}

\begin{figure*}[thb]
\centering
\vspace{-.1in}
\includegraphics[width=5.5in, trim={.1cm .22cm .1cm .22cm},clip]{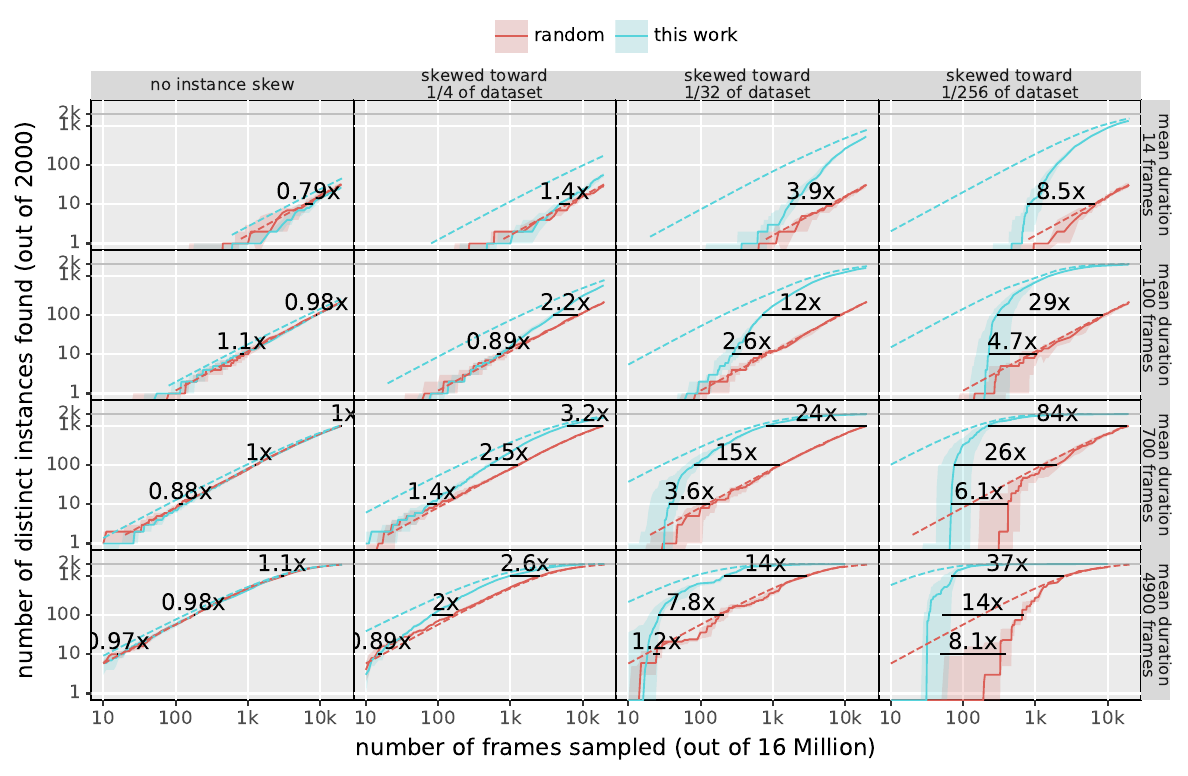}
\caption{\hl{Simulated savings in frames from \sys range from 1x to 84x, depending on instance skew (increasing skew from left to right, starting with no skew), and on the average duration of a instance in frames (increasing from top to bottom).  Solid lines show the median frames sampled by \sys and random. Shaded areas mark 25--75 percentile. We label with text the savings in samples needed to reach 10, 100, and 1000 results. The dashed blue line shows the expected results if samples were allocated optimally based on perfect prior knowledge of chunk statistics. The dashed red line shows the expected results from random. For a complete explanation see Section \ref{sec:simulation}.}
}
\label{fig:simplot}
\end{figure*}

This section explores how different parameters, in particular skew in instances and average duration of instances,  affect \sys performance in simulation.  
%Our goal in this section is to show how \sys performs compared to \random\ while we control different dataset characteristics: 1) skew in how instances are distributed, 2) the average duration of a result instance.  Before we do that, we explain what we mean by instance skew.

Instance skew is a key parameter governing performance of \sys.  If we knew 95\% of our results lie on the first half of the dataset, then we could allocate all our samples to that first half. This  would boost the $p_i$ of those results in the first half by $2\times$,  without requiring knowledge of their precise locations. Hence, we could expect about a $2\times$ savings in the frames needed to reach the same number of results. Skew arises in datasets whenever the occurrence of an event correlates relevant latent factor that varies across time e.g., time of day or location (city, country, highway, camera angle) where video is taken. \sys exploits this skew via sampling  when it exists in the data, even though \sys is not aware of these correlated factors or expects them as inputs.
	
%\subsection{Rarity} 
%
%A second factor affecting \sys is how rare classes really are. In the two extreme cases of extremely rare or extremely common objects, \sys will perform similar to 
%
%	
%To study the effect of skew, we performed a simulation where we allocated $1/S$ of the dataset that would hold 95\% of the instances, for different values of $S$.

In our simulation we show that under a wide variety of situations \sys can outperform \random\ by 2x to 80x, and it never performs significantly worse. We also explore in detail situations where  \random\ is competitive with \sys.  Results are shown results in \autoref{fig:simplot}.

	To run the simulation, we fixed the number of instances N to 2000, and the number of frames to 16 million. 
	We placed these 2000 instances according to a normal distribution centered in these 16 million frames, varying the standard deviation to result in no skew (left column in the \autoref{fig:simplot}) and skew where 95\% of the instances appear in the center 1/4, 1/32, and 1/256 of the frames (additional columns in the figure).
% 	To fully describe a set of instances in that data we can choose N=2000 center positions, and their durations. To control how the centers are spread out within those 16 million frames we use a normal distribution with its center somewhere in the middle of the 16 million frames,  and we change its standard deviation to adjust skew. For no instance skew, we use a standard deviation much larger than the 16 million frames (left column of figure \autoref{fig:simplot}).
% 	To skew results to mostly fit within 1/4th of the dataset (2nd column of figure), we let $4\sigma = 16/4 \text{ million}$, which forces 95\% of the samples to fall under that spread. We proceed similarly for instance skew of $1/32$ and skew of $1/256$ (in the 3rd and 4th column).  
To generate varied durations for each of these instances, we use a LogNormal distribution with a target mean of 700 frames. This creates a set of durations where the shortest one is around 50 frames and the longest is around 5000 frames. To change the average duration we simply multiply or divide all durations by powers of 7, to get average durations of $700/49 = 14$, $700/7 = 100$, $700$, $700\times 7=4900$.  These correspond to the rows in \autoref{fig:simplot}. 
	
	Once we have fixed these parameters, the different algorithms proceed by sampling frames, and we track the instances that were found (y axis of subplots) as we increase the number of samples we processed (x axis). For \sys, we divide the frames into 128 chunks. 
	%An experiment run creates a trajectory curve that starts at the origin and moves up and to the right, sometimes reaching 2000 on the y axis, or the simulation stops at 20k. 
	We run each experiment 21 times. In \autoref{fig:simplot} we show the median trajectory (solid line) as well as a confidence band computed from the 25th to the 75th percentile (shaded). Additionally, the dashed lines show the expected $N$ following the optimal allocation for each $n$, computed using the formula in \autoref{eq:optimal}.  \sys and \random\ start off similarly, which is expected because \sys starts by sampling uniformly at random. As samples accumulate \sys's trajectory moves more and more toward the dashed blue line which represents the expected $N$ if the allocation of samples across chunks was optimal from the beginning. 
	%\autoref{fig:simplot} shows that qualitatively speaking, \sys will match or exceed \random, depending on whether there was some skew in the data to exploit. 
We can see \sys generally outperforms \random, but performs similarly in two cases: 1) when there is very little skew in the locations of results (top row of figure) and 2) when results are very rare, which makes getting to the first result equally difficult for both (left column of figure).

%\srm{Now show \autoref{fig:simplot}, and explain it.  Don't worry about math below.}

% For practical reasons, in this paper we will quantify instance skew for a dataset as the fraction $1/S$ of the dataset frames that holds 95\% of the instances. Since strictly speaking we could pick a single frame per instance, and cover all 100\% of the instances,  we must qualify this definition of instance skew by assuming the data is partitioned into chunks of at least some size, and that $1/S$ is the fraction of {\em chunks} that hold 95\% of the results.  In the next section we will quantify how varying the exact size of chunks affects \sys.  This is the definition of instance skew used in \autoref{fig:simplot}.

\if{0}
\TODO{these equations inspired the parameters I used for the simulation, but I'm not finding it easy to weave them in}
\begin{eqnarray*}
	R(n+1) \gtrapprox N\left[ \mu_p - n(\mu^2_p +\sigma^2_p)\right]
\end{eqnarray*}

\if{0}
\begin{eqnarray*}
\max_j N_j\left[ \mu_{pj} - n_j(\mu^2_{pj} +\sigma^2_{pj})\right]
\end{eqnarray*}
\fi

\begin{eqnarray*}
.95SN\left[ \mu_{p} - n_j S(\mu^2_{p} +\sigma^2_{p})\right]
\end{eqnarray*}

\TODO{what about duration skew $\sigma_p$}
Takeaways here: by restricting sampling to the subset with the skew, we increase the expected payoff of the first sample to $.95SN\mu_p$ from $N\mu_p$, but also affect how fast the rate of new results decays over time (the new S factor in the negative term), it drops $S\times$ as fast. If we split that further into chunks, it suggests you need to also be ready to switch to sampling other locations. But having picked .95 as the skew definition makes it hard to make this argument flow. Will come back to this issue later.

\fi

% We highlight that is only a sufficient condition in which one can get a speedup from weighted sampling, not a required condition. We would also see an advantage from \sys if eg $\mu_p$ varies from chunk to chunk, even if we hold $N$ constant across chunks, and there is no instance skew, but a systematic duration variation.
% \todo{it might be good to check the results and see if this is also happening, first explain the results and then think of this.}

\if{0}
Intuitively, we expect $R(n)$ to grow proportionally to $N$, and to $\mu_p$. This is exactly true early on: the first sample drawn has an expected return of $R(1) = N\mu_p$ instances. For larger $n$, this intuition greatly overestimates the number of results we get. This can happen even as early as with $n=2$. For example, with two objects $p_1 = p_2 = .5$, $R(1) = 2\times.5 = 1$, and $R(2) = .5(1 -.5) + .5(1-.5) = .5$. In contrast, still with $N=2$ and $\mu_p=.5$ but with $p_1 = 1, p_2 \approx 0$, we have $R(1)$  remains $1$, but $R(2) \approx 0$. In both cases $N$ and $\mu_p$ are the same, but in the high skew case $R$ becomes 0 after only one sample.  We formalize this intuition by providing a bound on $R$ that accounts for both $\mu_p$ and $\sigma_p$ skew in the probabilities.

This bound might seem loose in the sense that as $n$ grows the right-hand side becomes negative, which is clearly an underestimate, but is tight for $n = 1$ and $n=2$.  For $n=1$, the bound becomes $N\mu_p$ as before, and for $n=2$, the bound is $N[\mu_p - \mu_p^2 - \sigma_p^2]$. In the examples above, with $p_1 = 1 - \epsilon$ and $p_2 = \epsilon$, $\sigma^2_p = 2*.5^2/2 = .25$ and $\mu_p = .5$, so $\mu_p^2 = .25$, and $R(2) = 2[.5 - .25 - .25] = 0$. It intuitively tells us that a sufficient condition for $R(n)$  not to drop very quickly is for $\sigma_p^2$ to be small. For example, if all $p_i$ are equal (e.g., all events have equal duration), then finding any one $p_i$ will cause $R(n)$ to drop by a constant $\mu_p = p_i$, and this will happen with probability $p_i=\mu_p$, so we expect a decrease of $\mu_p^2$. If the $p_i$ are not constant, then with probability $p_i$ $R(n)$ will drop by $p_i$, and the larger $p_i$ are much more likely to be the ones sampled, which is where the $\sigma^2_p$ term comes from.

If there is a fraction 1/$F$ of the dataset where most of the $N$ results concentrate, then we can also apply the above equation to a hypothetical case where 

when splitting the data into chunks, we can apply similar reasoning on each chunk. Splitting into chunks is advantageous when 

is large compared to

The effect of sampling randomly within a chunk rather than from the whole dataset are twofold, on the one hand, the duration of an event as a fraction of the chunk size becomes larger. $\mu_{pj} = M\mu_p$ (same with $\sigma_p$). On the other hand, the total number of results may drop. If $N_j$ is $N/M$, there is no net gain. On the other hand, if there is some M for which some chunk has most of the instances, i.e., $N_j > .9 N$, then we get

This means that we have boosted the payoff by at most $M$ early on, 
and will accumulate these instances after only $n/M$ samples. 

So, at least initially when $n$ is small, there are large differences between $N_j$ across chunks.
 This is likely to happen when there is some underlying relation between the object frequencies and durations within a chunk and some latent factor such as location, weather, time of day, season, etc.  This advantage is sustained longest when the average durations $\mu_{pj}$ are small and the skew in durations $\sigma_{pj}$ is small as well (we are less likely to find repeat results).

Note that this reasoning applies in the case of random, the $p_i$ arise from the number of frames in which object $i$ appears. However, this reasoning also applies to any weighted sampling strategy, or one based on scoring frames ahead of time and sampling proportional to that score: all these  can be thought of as inducing a set of $p_i$, and hopefully maximizing a $\mu_p$. The above  tells us that increasing $\mu_p$ while increasing $\sigma_p$ can mean we find fewer new examples even in the short run.
\fi

\subsection{Number of chunks}
\label{sec:chunks}
%In our experiments we use files a natural chunk boundaries as well 20 minute intervals for splitting longer videos internally, and this choice seems to perform well enough.

\begin{figure}
\centering
    \centering
    \includegraphics[trim={.25cm  .2cm .25cm .2cm},clip,scale=.9]{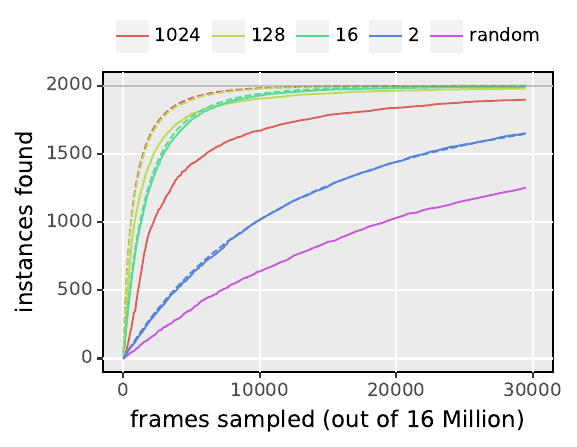}
\caption{Varying the number of chunks for a fixed workload. For 2 and 16 chunks, the dashed lines corresponding to optimal sample allocation using
\autoref{eq:optimal} match closely with results using \sys on the simulated data. For 128, and even more so for 1024 chunks, there is a noticeable gap between the dashed and solid lines.
%More chunks can exploit more skew as in the top \autoref{fig:chunks}, but they also demand more samples devoted to exploration, making convergence to optimal slower as shown in the bottom \autoref{fig:gap}.
}
\label{fig:chunks}

% {Increasing chunks shows decreasing returns (see dashed lines). The difficulty of converging to the optimal allocation increases with the number of chunks. The vertical gap between solid lines and dashed lines in the first subfigure is shown in the second subfigure.}
\vspace{-.2in}
\end{figure}

The way we split the dataset into chunks is an external parameter that affects the performance of \sys. We consider two extremes:
% In fact, the number of chunks must not be too few or too large, and to understand this it helps to consider two extreme scenarios.
 1) Suppose we have a single chunk. This makes \sys equivalent to random sampling.  Fewer chunks reduce \sys ability to exploit any skew. For example, the extreme case using two chunks imposes a hard ceiling on savings of $2\times$, even if the skew is much larger. 2) On the opposite extreme, one chunk per image frame would also be equivalent to random sampling but for a different reason: In this case, picking a chunk is equivalent to picking a frame.  Since there is no data initially about any chunks, Thompson sampling is equivalent to random sampling. Therefore, the extreme cases are clearly not practical and the choice of chunks can affect savings in principle. 
 In this section we explore this problem through simulation and we find we can vary the number of chunks across three orders of magnitude, while still remaining efficient. 
  
 In \autoref{fig:chunks} we show how \sys behaves for a range of chunk scenarios spanning several orders of magnitude. In this simulation, we fix instance skew as defined before to 32 and $p_i$ to a mean duration of 700 frames, the same values as the third column and third row of \autoref{fig:simplot}. We vary the number of chunks from 1 to 1024. Dashed lines show the number of instances found when using (static) optimal weights from \autoref{eq:optimal}, computed as a function of the x axis. The more chunks there are, the steeper the dashed lines because we have exact knowledge of frame distribution and can exploit skew at smaller time scales. 
Solid lines in \autoref{fig:chunks} show  the median performance of \sys under the same chunk and data settings, but without knowledge ahead of time. Notably, increasing the number of chunks can decrease performance: for example, going from 128 to 1024 chunks shows a decrease, even though both have similar optimal allocation (dashed) curves. As we increase the number of chunks $M$, we also increase the number of samples \sys needs to take just to be able to tell which chunks are more promising: for example, when working with 1024 chunks, just for each chunk to be sampled once we would need to allocate 1024 samples, and this by itself would be too little information about which chunks are more promising. We note that in our simulation we varied the number of chunks by three orders of magnitude and still see a benefit of chunking versus random across all settings; however, the benefits are non monotonic.

\if{0}
\subsection{Generalizations}

{\bf{Generalized instance durations}} Throughout the paper we have  used $p_i$ as proxy for duration, assuming we select frames with uniform random sampling. However, we could weigh frames non-uniformly at random in the first place, for example by using some precomputed per-frame score. If we use a non-uniform weight to sample the frames, we effectively induce a different set of $p'_i$ for each result, ideally one with $\mu_{p'} \gg \mu_p$. The estimates for the relative value of different chunks will still be correct since \sys\ is designed to work with any  underlying $p_i$.

{\bf{Generalized chunks}} We have introduced the idea of a chunk as a partitioning of the data. However, it would be possible for chunks to overlap, and this is equivalent to having instances that span multiple chunks. This choice is only meaningful if the chunks are different in some way, for example one chunk can be a half hour of video sampled uniformly at random, while a different chunk can be the same half hour of video sampled under a different set of weights, using the idea of generalized $p_i$. 
\fi

\section{Evaluation} \label{sec:eval}

Our goals for this evaluation are to demonstrate the benefits of \sys\ on real data compared to alternatives, including random sampling and proxy models. We show that, on these challenging datasets, \sys\ achieves savings of up to 6x over random sampling. We also show that by avoiding the large upfront cost of scoring every frame in the dataset with a query-specific proxy model, \sys/ is able to find results quicker.

\subsection{Experimental Setup}

\noindent
\textbf{Baselines.}
We compare \sys against two of the baselines discussed in Section \ref{sec:objquery}: 1) uniform random sampling and 2) state-of-the-art representative of the proxy model idea (BlazeIt)~\cite{blazeit}. BlazeIt is a  state-of-the-art video processing system that processes several different types of queries, including queries with more complex predicates and aggregates. In this evaluation we are only concerned with how a system like BlazeIt, as a representative of proxy-based approaches, handles the specific case of distinct object queries.

\medskip
\noindent
\textbf{Implementation.}
Random sampling: BlazeIt and \sys are at their core sampling loops where the choice of which frame to process next is based on an algorithm-specific decision. Random simply chooses the next frame randomly. \sys chooses based on \autoref{alg:pseudocode}. Proxy-based approaches such as BlazeIt choose the frame with the next highest score;  \highlight{these scores are pre-computed by running a query-specific proxy model over every frame of the dataset.} \hlinelabel{line:score}
	
We implement this sampling loop in Python, using PyTorch to run inference on a GPU. For the object detector, we use Faster-RCNN with a ResNet-50 backbone. To achieve fast, random access frame-decoding rates we use the Hwang library from the Scanner project\cite{scanner}, and re-encode our video data to insert keyframes every 20 frames.

%	We implement the subset of \blazeit\ that optimizes distinct object limit queries, based on the description in the paper~\cite{blazeit} as well as their published code. We opted for our own implementation to make sure the I/O and decoding components of both \sys\ and \blazeit\ were equally optimized, and also because extending BlazeIt to handle our own datasets and metrics was infeasible. For BlazeIt's cheap proxy model, we use an ImageNet pre-trained ResNet-18. This model is more computationally expensive than the  used in that paper, but also more accurate. Note that our timed results do not depend strongly on the inference cost of ResNet-18, since the model is still sufficiently lightweight that BlazeIt's full scan is dominated by video IO and  decoding time.

% 5. Page 9, Dataset paragraph, line 2, bdd should be BDD.
% 6. Page 9, Datasets paragraph, line 8, becaue should be because.

\medskip
\noindent
\textbf{Datasets.} \hlinelabel{line:datasets} We use six different datasets in this evaluation, which we call dashcam, BDD-1k, BDD MOT, archie, amsterdam, and night-street.  The dashcam dataset consists of 10 hours of video, or over 1.1 million video frames, collected from a dashcam over several drives in cities and highways. Each drive can range from around 20 minutes to several hours. \highlight{Drives longer than 20 minutes are split into 20 minute chunks.} The BDD dataset used for this evaluation consists of a random sample of 1000 random video clips from the Berkeley Deep Drive Dataset\cite{bdd}. Because BDD video clips are less than one minute long, \highlight{we are forced to use each small clip as an individual chunk for a total of 1000 chunks}. This constraint makes it a challenging scenario for \sys, as we explained in \autoref{sec:analysis}. \highlight{ BDD MOT (multi-object tracking) is a different subset of the BDD dataset of 1600 short clips of about 200 frames each. Unlike the other datasets, it manually labelled ground truth, including instance ids.} 
The other three datasets: archie, amsterdam, and night-steet are video from fixed cameras overlooking urban locations. Each consists of 20 hours of video. These datasets are used in the evaluation of related work on optimizing (other) types of video queries on static camera datasets~\cite{blazeit,focus,chameleon}. \highlight{We also use 20 minute chunks for these datasets, for a total of about 60 chunks.} There are two important differences between static camera and moving camera settings: 1) In the moving camera setting there is more room for instance skew depending on how the scenery changes over time (for example, the BDD dataset includes data from multiple cities and countries, as well as from highway, suburban and urban settings). 2) For proxy model based approaches, the moving camera setting is also more challenging to create an accurate proxy model for because the background changes constantly and the inputs to the proxy model vary more widely in appearance.\\
\hlinelabel{line:groundtruth} 
\noindent{\bf Queries and ground truth} We picked between six and eight objects for each of the datasets based on what they show.  None of the datasets have human-generated object instance labels that identify objects over time.  Therefore, we approximate ground truth by sequentially scanning every video in the dataset and running each frame through a reference object detector (we reuse the Faster-RCNN model that we apply during query execution). If any objects are detected, we match the bounding boxes with those from previous frames and resolve which correspond to the same instance. To match objects across neighboring frames, we employ an Intersection over Union (IoU) matching approach similar to SORT~\cite{SORT}. IoU matching is a simple baseline for multi-object tracking that leverages the output of an object detector and matches detection boxes based on overlap across adjacent frames. For all datasets we found we needed to fine-tune the object detector and the IoU matching to get to a reasonable quality ground truth.  This process involved manually annotated around a hundred frames for each of archie, amsterdam, and night-street, retraining Faster-RCNN and then running it over the full datasets to get better quality ground truth. For BDD and dashcam we fine-tuned Faster-RCNN with the BDD labels. For these three datasets we use object types, including but not limited to those used in \cite{blazeit}. \\
In addition to searching for different object classes, we also vary the limit parameter recall of 10\%, 50\% and 90\%, where recall is the fraction of distinct instances found.  These recall rates are meant to represent different kinds of applications:  10\% represents a scenario where an autonomous vehicle data scientist is looking for a few test examples, whereas a higher recall like 90\% would be more useful in an urban planning or mapping scenario where finding most objects is desired.

\subsection{Results: sampling immediately vs. proxy scoring overhead}
\label{sec:scantime}

The main advantage of \sys with respect to using a proxy model that first scores frames is that \sys requires no prior scoring phase, and therefore avoids a costly scan over the full video dataset. In   \autoref{tab:times} we show for all queries \sys will reach .9 recall before a proxy model has finished scoring the dataset. \sys reaches .1 and .5 recall orders of magnitude faster because the proxy model needs to score everything before it offers results, as it relies on returning the highest scoring frames to maximize its accuracy. 

\begin{table}
\smaller
\center
\begin{tabular}{lllrrr}
\toprule
             &    &  (\% instances)     &   10 &    50 &    90 \\
        & proxy &  & & & \\
dataset & (scan) & category &       &        &        \\
\midrule
BDD 1k & 54m & bike & 1m37s &  8m57s &    41m \\
             &    & bus & 1m17s & 10m38s &    49m \\
             &    & motor & 1m38s &  8m53s &    46m \\
             &    & person &   52s &  6m46s &    36m \\
             &    & rider & 1m31s & 10m14s &    45m \\
             &    & traffic light & 1m33s & 12m18s &    50m \\
             &    & traffic sign & 1m38s &    14m &    58m \\
             &    & truck &  1m8s & 10m39s &    50m \\
 \midrule      
BDD MOT & 53m & bicycle &   52s &  6m51s &    35m \\
             &    & bus &   31s &  3m18s &    21m \\
             &    & car & 1m31s &  8m21s &    30m \\
             &    & motorcycle &   49s &  6m38s &    39m \\
             &    & pedestrian &   41s &  4m51s &    24m \\
             &    & rider &   59s &  6m17s & 32m50s \\
             &    & trailer &   37s &  3m54s &    38m \\
             &    & train &   18s &     3m &    32m \\
             &    & truck &   36s &  3m57s & 20m36s \\
\midrule      
amsterdam & 9h50m & bicycle & 1m10s &  8m42s &    39m \\
             &    & boat &    2s &    14s &     4m \\
             &    & car &   45s &     7m & 23m33s \\
             &    & dog & 1m51s & 12m46s &  1h49m \\
             &    & motorcycle & 5m21s & 24m58s &  2h18m \\
             &    & person &   29s &  4m20s & 21m39s \\
             &    & truck &   46s &     9m &    39m \\
 \midrule      
archie & 9h49m & bicycle &  1m4s &     8m &    43m \\
             &    & bus &    1m &  6m47s &    58m \\
             &    & car &   46s &  4m36s & 10m35s \\
             &    & motorcycle & 3m10s &    22m &  1h57m \\
             &    & person &  1m5s &  7m32s &    50m \\
             &    & truck & 1m36s & 13m41s &  1h21m \\
 \midrule       
dashcam & 2h54m & bicycle &   32s &  5m38s &     1h \\
             &    & bus & 1m11s &    26m &  2h58m \\
             &    & fire hydrant & 1m40s &    16m &  1h15m \\
             &    & person &   20s &  4m22s &   1h8m \\
             &    & stop sign &   45s & 20m26s &  2h27m \\
             &    & traffic light &   26s &     7m &  1h21m \\
             &    & truck & 2m17s & 28m37s &  2h58m \\
 \midrule       
night street & 8h & bus & 1m27s &  9m55s &    52m \\
             &    & car &   12s &  2m21s &    11m \\
             &    & dog & 2m34s & 18m45s &  3h39m \\
             &    & motorcycle & 9m13s &  1h52m &  7h31m \\
             &    & person &   14s &  1m55s &    15m \\
             &    & truck & 1m10s &  9m59s &   1h4m \\
\bottomrule
\end{tabular}
\caption{Time for the scanning component of a proxy based approach vs. time taken by \sys. Across all queries and datasets, it is cheaper to reach 90\% of instances using \sys sampling than it is to scan and score frames prior to sampling, and much easier to reach 10\% and 50\% of instances. A separate comparison between \random and \sys is shown in \autoref{fig:ratioplot} }
\label{tab:times}
\end{table} 

The totals in \autoref{tab:times} were computed by measuring the scoring throughput we can sustain on our equipment (100 frames per second, bound by io+decode), and then estimating the time it would take to fully scan the dataset. For sampling, we measure how  many frames \sys needs to sample and process in that time and check how many results it would find, knowing from measurements that \sys processes frames at a rate of 20 frames per second, bound by the object detector throughput. Proxy scoring times reach up to 10 hours; this is the reason we prefer to estimate the time rather than implement and run it fully for each of the 34 queries we evaluate, as each would take hours to run.

\subsection{Time savings vs. recall level}
\label{sec:expresults}

The previous section showed the time it takes us to compute proxy scores for a whole dataset can be more than enough time for a weighted sampling like \sys to find most results in the data. Now we are interested in showing how much time \sys saves on real data to reach different levels of recall compared to \random, neither of which has a proxy score computation overhead, and hence can produce results immediately. We show results for three levels of recall over instances: .1, .5 and .9 in \autoref{fig:ratioplot}. The maximum is around 6x (leftmost) and the worst case (rightmost) is .75x for boats. The .9 percentile over the 100 bars in the plot is 3.7x and the .1 percentile is 1.2. The geometric mean of savings overall is 1.9 across all vertical bars.

%0.1    1.189272
%0.5    1.692983
%0.9    3.724493 

\begin{figure}
% \begin{subfigure}{.601\linewidth}
%   \centering
%   \includegraphics[width=\linewidth]{assets/timeplot.pdf}
%   \caption{Times}
%   \label{fig:sub1}
% \end{subfigure}%
% \begin{subfigure}{.399\linewidth}
  \centering
  \includegraphics[width=\linewidth]{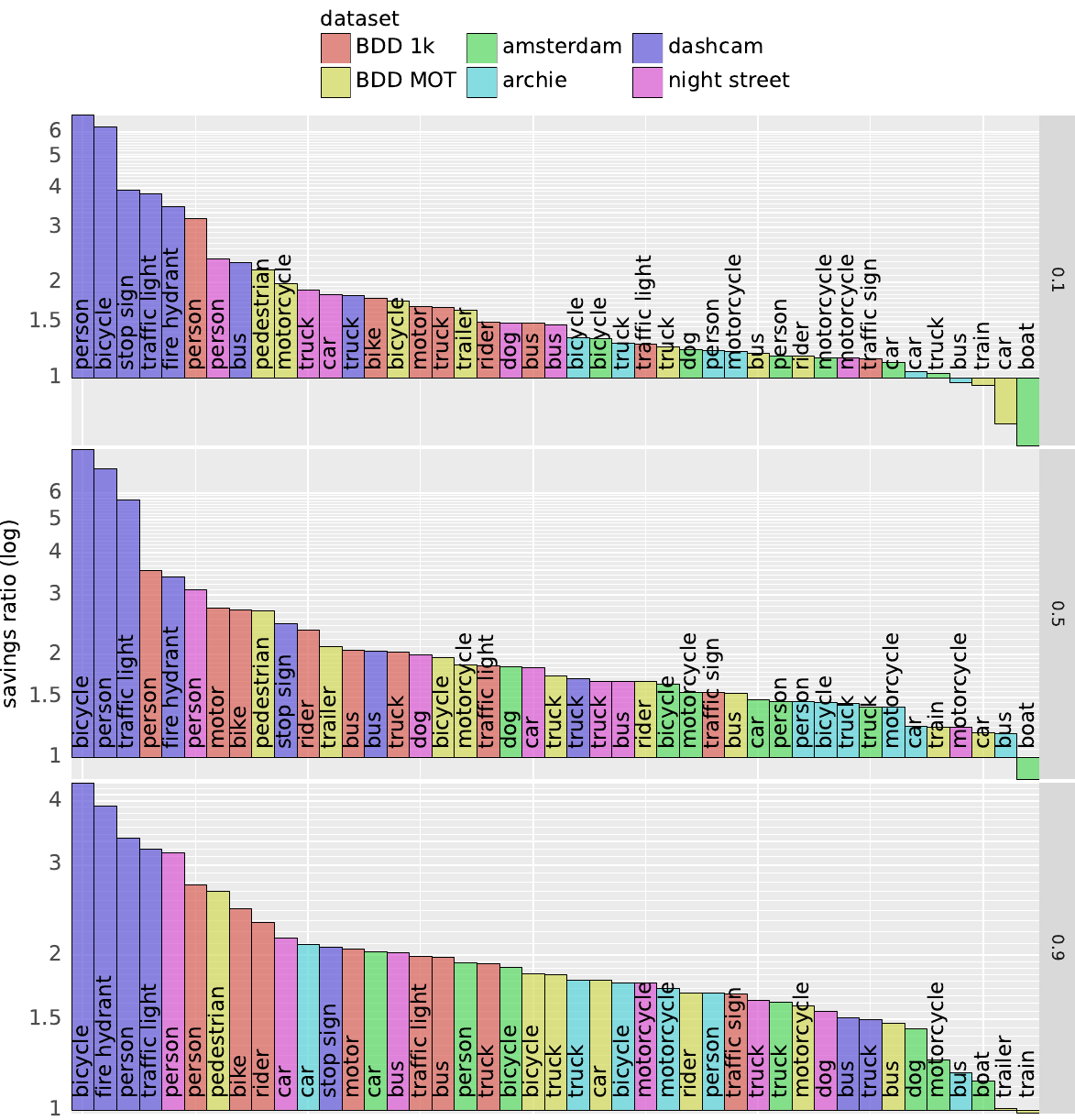}
%   \caption{Savings ratio}
%   \label{fig:sub2}
% \end{subfigure}
\caption{Time savings ratio when using \sys vs. \random\ for all queries. The top panel corresponds to time savings to reach .1 of all instances, the middle panel to time savings for reaching .5, and the bottom panel to reach .9. }
\label{fig:ratioplot}
\end{figure}

\autoref{fig:skewplot} shows what the chunks for the more extreme cases of \autoref{fig:ratioplot}  look like in terms of skew and abundance. For example, bicycle on the dashcam dataset shows very large skew, which explains why it can achieve results of 3.7x. Next we have motorcycle in BDD, which even though has very large skew, has relatively low savings of 2x. The reason for the moderate gain is the large amount of chunks in the BDD dataset; it takes a while for \sys to sample all of them, so it takes a while to notice and take advantage of the skew. \autoref{fig:ratioplot} shows the bar for motorcycle (``motor'' under the BDD color) reaches a savings of 3x at recall .5, but has savings of less than 2x early on, likely because it takes time to identify the skew. Finally, ``boat'' in the amsterdam dataset, the worse performing query, and ``car'' in the archie dataset have very low skew, meaning \random should do just as well.  The person in the night-street (aka town-square) dataset has moderate amounts of skew and \sys is able to exploit it.

\autoref{fig:ratioplot} shows \sys produces robust gains across 5 datasets and a total of 34 queries, while \autoref{tab:times} shows that even a perfect proxy adds too much overhead. Therefore, by the time we have scored frames for a new object class, we could have already sampled a lot of the individual instances. Part of the reason this is possible is because even the scoring model scores every frame; whereas, \sys samples individual frames and avoids sampling frames likely to return existing results. If individual instances last longer, then \sys will naturally get a large number of instances without needing to look at too many frames. \random\ sampling also has this property, but \autoref{fig:ratioplot} shows \sys will out-compete random sampling as well.

% 8. Page 11, Figure 7, the last line of the caption, something is missing after Section.

\begin{figure}[htb]
\includegraphics[width=3.3in, trim={.2cm .5cm .25cm .3cm},clip]{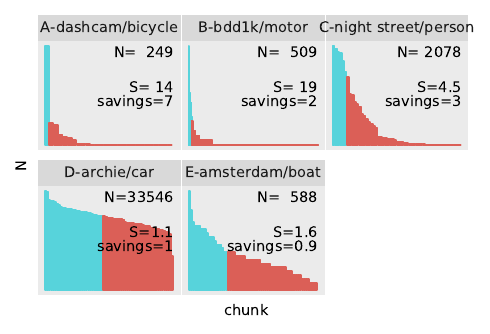}
\caption{Instance skew and savings for a few representative queries from \autoref{fig:ratioplot}.
Each vertical bar on these plots corresponds to a chunk, and its height is proportional to number of instances. Blue colored bars are the minimum set of chunks that cover half the instances. $S$ is our skew metric defined in \autoref{sec:simulation}}
\label{fig:skewplot}
\end{figure}

\if{0}
\subsubsection{Dashcam Dataset Results}
Overall, \sys\ saves more than $3\times$ in frames processed vs. \random, averaged across all classes and different recall levels. Savings do vary by class, reaching savings of above 5x for the bicycle class in the dashcam dataset. In terms of frames we can see \blazeit\ performs really well at recall of 10\% in the dashcam dataset, reaching $10\times$ the performance of \random. The reason is that in the dashcam dataset  \blazeit\ proxy models learn to classify frames much better than random, as is evident in \autoref{fig:average_precision}. For high recall levels, however, this initial advantage diminishes because it is difficult to avoid resampling the same instances in different frames, despite skipping a gap of 100 frames when a result is successfully found.

The story becomes more complex when comparing runtime. On the right panel of \autoref{fig:dashcam} we show the total runtime costs, including overhead incurred prior to processing the query, which erase the early wins from better prediction on the left panel. The impact of this overhead reduces in magnitude as more samples are taken, and so the difference is less dramatic at higher recall levels.
For \blazeit\, the total time in the plot is a sum of three separate costs: training, scoring, and sampling. Training costs are mostly due to needing to label a fraction of the dataset using the the expensive detector. Scoring costs come from needing to scan and process the full dataset with the proxy model. Finally, sampling costs come from visiting frames in score order.% (or random order, in the case of \random).

In \autoref{tab:xput} we quantify the costs of the three steps, where for BlazeIt we estimate sampling costs by assuming a well-trained proxy model finding instances $100\times$ faster than random. The table shows that BlazeIt's runtime is dominated first by scoring (i.e., the full scan to process every frame through the proxy model), and second by labeling. While it may be possible to amortize the costs of labeling in BlazeIt across many queries, because the expensive model may provide general-purpose labels, the full scan that dominates the overhead cannot be amortized as the proxy is query-specific.
Because \random\ does not need to scan the entire dataset before yielding results, it can afford to be much lower in precision while still providing a faster runtime, especially at low recalls. ExSample outperforms the other methods because it provides high sampling precision while avoiding the need for a full scan.

\begin{figure}[htb]
\includegraphics[width=3.2in, trim={.cm .cm .cm .cm},clip]
{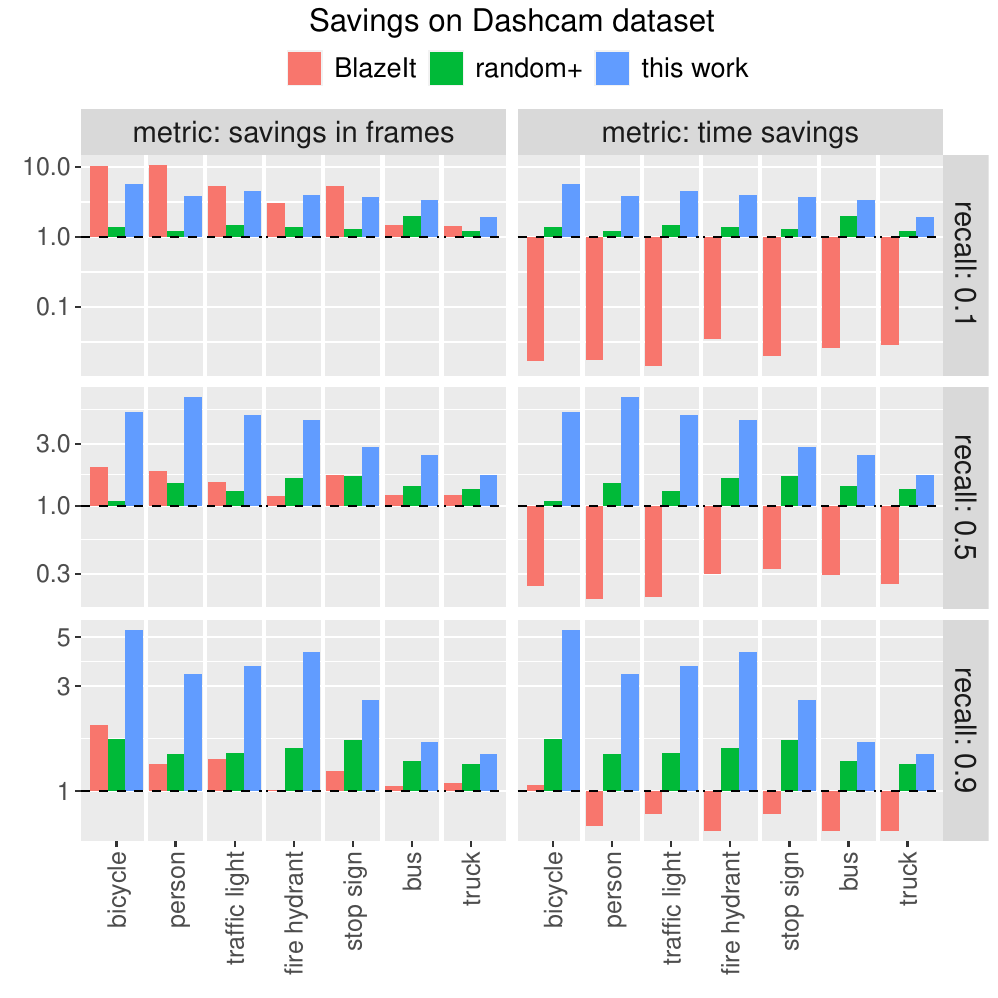}

\caption{Results on the Dashcam dataset. Comparing frames (left column) and times(right column) for different recall levels (each row) on the dashcam dataset. Different methods (color) are compared relative to the costs of using \random\ to process the same query and reach the same level of instance recall. 
%The right column includes overheads incurred prior to processing, broken down in \autoref{tab:xput} 
}
\label{fig:dashcam}
\end{figure}

\begin{table}[htb]
\vspace{-.2in}
\centering
\small
\begin{tabular}{lrrrrr}
\toprule
{} &  through- & \blazeit &   random  \\
{stage} &  put(fps) &  frames/time(s) &  frames/time(s) \\
\midrule
label      &               20 &         110K/5.8Ks &        0  \\
score &               100 &         1.1M/12Ks &             0 \\
sample   &                20 &        120/5.8s &            1.2K/58s \\
\bottomrule
\end{tabular}
\caption{\hl{Time breakdown: random vs \blazeit{}.}}
\label{tab:xput}
\vspace{-.2in}
\end{table}

\subsubsection{BDD Results}
The maximum benefit from \sys\ is lower in \autoref{fig:bdd}, reaching only 3x. BDD is a challenging scenario for \sys because as we have mentioned earlier, it is composed of 1000 40-second clips.  BDD is also challenging for \blazeit, because proxy models have a harder time learning with high accuracy on this dataset, which presents many points of view and small objects. \autoref{fig:average_precision} shows that proxy models do not reach particularly high Average Precision. This is not an issue for \sys because we do not rely on a cheaper approximation, which may not be as easy to get for highly variable datasets.

\begin{figure}[tbh]
\includegraphics[width=3.2in, trim={.cm .cm .cm .cm},clip]
{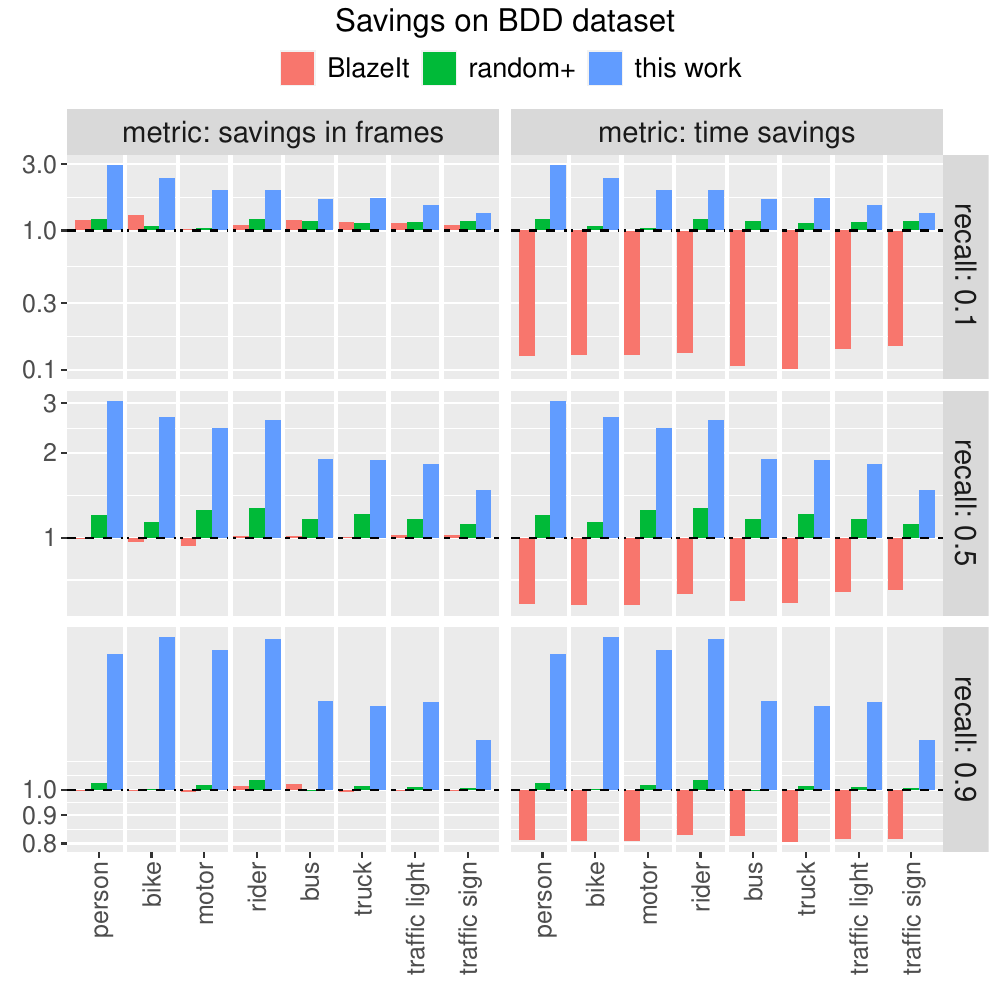}
\caption{Results on BDD dataset. Proxy has a harder time learning on this dataset, with lower accuracy scores. Counter-intuitively, this makes its savings comparable to those of \random\ sampling, improving its results compared to Dashcam datasest }
\label{fig:bdd}
\vspace{-.2in}
\end{figure}

\begin{figure}[htb]
\includegraphics[width=3.2in, trim={.cm .cm .cm .cm},clip]
{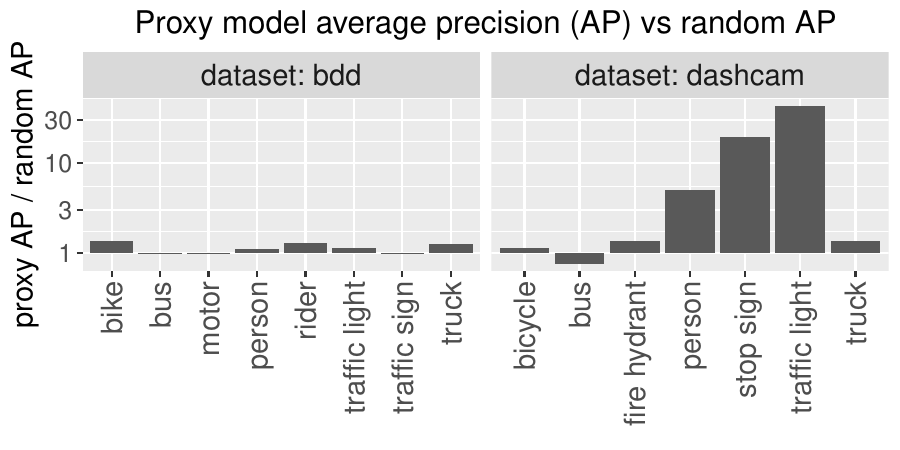}

\caption{Average Precision of proxy models vs. a randomly assigned score. For BDD, the proxy model does only slightly better than random in precision, which causes it to perform similarly to random when measured in savings in  \autoref{fig:bdd}. For the Dashcam dataset, the proxy models are more accurate than random.  Counter-intuitively, higher accuracy in scores hurts relative performance in  \autoref{fig:dashcam}  due to the greediness of the algorithm. }
\label{fig:average_precision}
\end{figure}

%\begin{figure}[!ht]
%\includegraphics[width=3.2in, trim={1.cm .1cm .1cm .cm},clip]{assets/xput_breakdown.pdf}
%\caption{Throughput of each processing phase in our implementation of BlazeIt.  The yellow boxes show the overall throughput reached by those processing phases in our implementation.
%Additionally, to distinguish the bottlenecks from the object detector from those of video decoding, we show the maximum throughput achievable  by I/O and video decoding in purple,  and inference in cyan. Inference with proxy models is marked to distinguish it from inference with the expensive model. For the scoring phase, which dominates the bulk of the overhead in BlazeIt, scanning through the dataset bounds the throughput to 100fps in our dataset. Although labelling throughput is low, labelling only happens on a fraction of the dataset, so represents a small fraction of the overall runtime.}
%\label{fig:xput}
%\end{figure}
BlazeIt\cite{blazeit} prioritizes sampling the highest scoring frames, where the score is computed with a cheaper proxy model. Answering queries with such systems involves four stages, whose throughputs on our data are shown in \autoref{fig:xput}.

	\begin{enumerate} 
	\item {\em labelling phase}: requires labelling a fraction of the dataset with the expensive object detector. Its runtime grows linearly with training set size, and because the object detector is involved, the throughput is as low as that of the sampling phase.
	\item {\em training phase}: once the labels are generated, a cheaper proxy model is fit to the dataset. This phase can be relatively cheap if the proxy is itself cheap and the training set fits in memory, avoiding any need for data I/O or decoding. \autoref{fig:xput} shows the throughput of ResNet-18 is indeed much higher and unlikely to be the bottleneck.
	\item {\em scoring phase}: the proxy model runs over the dataset, producing a score for each frame. Even if the proxy model is virtually free, \autoref{fig:xput} shows that the I/O and decode for the remainder of the dataset dominate the runtime.
	\item {\em sampling phase}: we fetch and process frames in descending order of proxy score. This phase ends when we have found enough results for the user. This is the only phase for \sys\ and for baselines such as \lsb. Regardless of access pattern, this phase is dominated by the cost of inference with the object detector.  
	\end{enumerate}
	
The first three phases can be seen as a fixed cost paid prior to finding results for proxy-based methods. The promise of these proxy-based techniques is that by paying the upfront cost we can greatly save on the rest of the processing.  But our results in the previous section show that \sys is often more effective than the proxy, without these high up-front costs.  Furthermore, these up-front costs have to be paid multiple times;  for example, if the user wants to look for a new class of object or process a new data set. It is unlikely that users will want to look for the same class of object on the same frames of video again, which pre-training a proxy makes more efficient.

\subsubsection{Additional Datasets}
\label{sec:newdatasets}
In addition to the Dashcam and BDD datasets, we evaluated \sys on archie, amsterdam, and night-street from the BlazeIt paper~\cite{blazeit}, and measured how \sys performs compared to random in terms of frame savings on all five datasets. This evaluation includes around 35 classes total across the datasets.
The results are aggregated below by dataset in \autoref{tab:bydataset} (averaging across recall levels) and by recall level in \autoref{tab:byrecall} (averaging across datasets).

\autoref{tab:bydataset} shows that savings are both query and data dependent. This is reflected in \autoref{fig:skewplot}, where we show an object for each dataset, and how its distribution is skewed across chunks in the dataset. The dataset/class combos with the largest skews (dashcam, bdd1k and night-street) also show the largest gains, as expected based on Section \ref{sec:simulation}. Amsterdam and archie have little skew overall on any class, hence \sys cannot perform much better than \random. On BDD, there are many short chunks due to the dataset structure, so the savings are lower than skew would suggest, as explained in Section \ref{sec:chunks}.

\begin{table}[htb]
\vspace{-.1in}
\centering
\small
\begin{tabular}{lcrrr}
\toprule
dataset       &  classes   &  min (savings) &  median &  max \\
\midrule
dashcam & 7 &          1.7 &             4.1 &          6.5 \\
bdd1k & 8 &          1.5 &             1.9 &          3.5 \\
night-street & 6 &          1.2 &             1.6 &          3.1 \\
amsterdam & 7 &          0.9 &             1.5 &          1.7 \\
archie & 6 &          1.2 &             1.4 &          1.5 \\
\bottomrule
\end{tabular}
\caption{\hl{Aggregate savings of \sys relative to random for the different datasets.}}
\label{tab:bydataset}
\vspace{-.1in}
\end{table}

Rather than running a specific proxy-based system like BlazeIt on these additional datasets, we compared \sys to an idealized proxy, which has a 100\% recall on identifying instances, so its cost is just the cost it takes to score all frames in the dataset.  We show how the fraction of instances that \sys can identify in the time it takes this idealized proxy do this scoring, using proxy timing data based on the experiments in \autoref{sec:expresults}, using the same five classes show in \autoref{fig:skewplot}.

The results are shown in \autoref{tab:vsproxy}.  From this, we can see that across a range of object classes and skew, \sys can detect most instances in a dataset before a proxy can score all frames.

\begin{table}[]
\vspace{-.1in}
    \centering
    \small
\begin{tabular}{llrr}
\toprule
            &      &  recall &  fraction \\
dataset & category & (this work)         & sampled                \\
\midrule
A-dashcam & bicycle &    0.72 &           0.07 \\
B-bdd1k & motor &    0.38 &           0.22 \\
C-night street & person &    0.90 &           0.02 \\
D-archie & car &    0.90 &           0.02 \\
E-amsterdam & boat &    0.90 &           0.02 \\
\bottomrule
\end{tabular}
    \caption{Recall level reached by \sys in the same amount of time it would take to scan the dataset. The experiment stops at recall .9.}
    \label{tab:vsproxy}
    \vspace{-.2in}
\end{table}

Finally, \autoref{tab:byrecall} compares \sys to random on all five datasets, 
aggregating by recall. Here \sys\ performs slightly better at 50\% recall, because at 10\% recall, with fewer samples, \sys\ has less time to identify chunks with more instances, and at 90\% recall, some of the instances are in less skewed chunks that \sys samples less.

\begin{table}[htb]
\vspace{-.1in}
\centering
\small
\begin{tabular}{lrrr}
\toprule
recall &  min (savings) &  median &  max \\
\midrule
10\%    &          0.8 &             1.4 &          6.9 \\
50\%    &          0.9 &             1.7 &          6.9 \\
90\%&          1.2 &             1.9 &          4.8 \\
\bottomrule
\end{tabular}
\caption{\hl{Aggregate savings broken down by recall.}}	
\label{tab:byrecall}
\vspace{-.2in}
\end{table}
\fi

\section{Related Work}
%\todo{Mention noscope and Blaze it are not necessarily targeting this kind of query. 2) Aaron Elmore has a paper in this space} 
% Visual data management systems such as DeepLens~\cite{deeplens} motivate the need 
\label{sec:relwork}
\noindent
\textbf{Video Data Management.}
There has recently been renewed interest in developing data management systems for video analytics due to the high cost of applying state-of-the-art video understanding methods, which almost always involve deep-neural networks, on large datasets. Recent systems adapt prior work in visual data management~\cite{qbic,chabot} for modern tasks that involve applying machine learning techniques to extract insights from large video datasets potentially comprising of millions of hours of video. In particular, DeepLens~\cite{deeplens} and VisualWorldDB~\cite{visualworlddb} explore a wide range of opportunities that an integrated data management system for video data would enable. DeepLens~\cite{deeplens} proposes an architecture split into storage, data flow, and query processing layers, and considers tradeoffs in each layer -- for example, the system must choose from several storage formats, including storing each frame as a separate record, storing video in an encoded format, and utilizing a hybrid segmented file. VisualWorldDB proposes storing and exposing video data from several perspectives as a single multidimensional visual object, achieving not only better compression but also faster  execution.

%\todo{below: for approximate models, we only really need to talk about blazeit since it's the only one focusing on limit queries. so if we're short on space, we can get rid of most of the discussion and just cite all the other approximate model papers in one sentence.}

\textbf{Speeding up Video Analytics.}
Several optimizations have been proposed to speed up execution of various components of video analytics systems to address the cost of mining video data. Many recent approaches train specialized proxy classification models on reduced dimension inputs derived from the video~\cite{noscope,msr,videomon,focus,miris}. As we detailed in Section \ref{sec:objquery}, most related to our work is BlazeIt~\cite{blazeit}, which adapts proxy-based video query optimization techniques for object search limit queries.

Video analytics methods employing proxy models are largely orthogonal to our work: We can apply ExSample to sample frames, while still leveraging a cascaded classifier so that expensive models are only applied on frames where fast, proxy models have sufficient confidence. The extensions to limit queries proposed to BlazeIt are not orthogonal, as they involve controlling the sampling method, but a fusion sampling approach that combines ExSample with BlazeIt may be possible. Nevertheless, a major limitation of BlazeIt is that it requires applying the proxy model on every video frame, even for limit queries; as we showed in Section \ref{sec:eval}, this presents a major bottleneck that can make BlazeIt slower on limit queries than random sampling.

Besides methods employing proxy models, several approaches, such as Chameleon~\cite{chameleon} as well as~\cite{videoedge,approxnet,vstore}, consider tuning optimization knobs such as the neural network model architecture, input resolution, and sampling framerate to achieve the optimal speed-accuracy tradeoff. Like proxy models, these approaches are orthogonal to ExSample and can be applied in conjunction.

An alternative approach altogether is to create an index of the dataset ahead of time. For queries on video, this could be done using a multi-use proxy embedding model. \cite{tasti} trains a convolutional model to be reusable for multiple tasks. Part of the motivation is to avoid the repeated per-query scan overhead that makes some of the approaches based on proxy models be very expensive; however, it requires some ahead of time knowledge of the queries that will be run.  \sys targets the problem of working with a dataset that has not been indexed ahead of time on ad-hoc queries, usable as long as there is a black-box detector to specify what is sought within the dataset.

\textbf{Speeding up Object Detection.} Accelerating video analytics by improving model inference speed has been extensively studied in the computer vision community. Because these methods optimize speed outside of the context of a specific query, they are orthogonal to our approach and can be straightforwardly incorporated. Several general-purpose techniques improve neural network inference speed by pruning unimportant connections~\cite{deepcompression,shrinknets} or by introducing models that achieve high accuracy with fewer parameters~\cite{mobilenets,squeezenet}.
Coarse-to-fine object detection techniques first detect objects at a lower resolution, and only analyze particular regions of a frame at higher resolutions if there is a large expected accuracy gain~\cite{gao2018dynamic,autofocus}. Cross-frame feature propagation techniques accelerate object tracking by applying an expensive detection model only on sparse key frames, and propagating its features to intermediate frames using a lightweight optical flow network~\cite{zhu2018,dff}.

\begin{highlightenv}

\section{Future Work}
\label{sec:futurework}
Two  areas for future work are automating chunking and integrating some benefits from scoring while avoiding scans of the dataset. 

For scoring. We note the equations in \autoref{sec:sys} remain valid even if sampling within a chunk is non-uniform but based on a score. The current downside of scoring frames is the scanning component; therefore, a key to integrating these approaches would be a form of predictive scoring of frames that avoids scanning.
\end{highlightenv}

\section{Conclusion}

In this paper we introduced \sys, a method for processing distinct object-search queries on video repositories through chunk-based adaptive sampling. The aim of the approach is to find frames of video that contain objects of interest, without running an expensive object-detection algorithm on every frame. \sys approaches the sampling problem by partitioning the data into chunks and adjusting the frequency of samples for each chunks based on the rate at which {\em new} objects are sampled from each chunk.  We formulate this sampling process as an instance of Thompson sampling,  and explicitly estimate the probability of finding a new object in a future random frame for each chunk.  As results in a particular chunk are found, \sys allocates more samples to that chunk; and as new results are exhausted, \sys naturally refocuses its sampling on other less frequently sampled chunks. Our evaluation of \sys on real world datasets shows it provides consistent savings without requiring an expensive scan of the data.

\if{0}
\section*{Acknowledgements}

Research was sponsored by the United States Air Force Research Laboratory and the United States Air Force Artificial Intelligence Accelerator and was accomplished under Cooperative Agreement Number FA8750-19-2-1000. The views and conclusions contained in this document are those of the authors and should not be interpreted as representing the official policies, either expressed or implied, of the United States Air Force or the U.S. Government. The U.S. Government is authorized to reproduce and distribute reprints for Government purposes notwithstanding any copyright notation herein.
\fi

%In particular, these proxy-based methods are much slower because they require running the proxy model on all frames.

%%% number of chunks...

%\begin{eqnarray*}
%	S(n) \leq N \left[1 - (1 - \mu_p)^n\right]
%\end{eqnarray*}

\if{0}
\section{optimality}
\TODO{ignore this section}
When using a fixed weighting $\vec{w}$ for each partition ahead of time,
the expected number of seen instances after $n$ samples are drawn randomly is $R(n,\vec{w}) = N \sum_i^N (1 - \vec{w}\vec{p_i})^n$.  Conceptually, the optimal $\vec{w^*}$ chosen ahead of time for the given data and partitioning can be seen as a minimization problem over the number of missing instances $M(n,\vec{w}) = \sum (1-\vec{wp_i})^n$. $\vec{p_i}$ is the vector for instance $i$ in our matrix $\mat{P}$, seen as a column. 

Note that in practice we do not know $\mat{P}$, which would be equivalent to knowing all the $p_{ij}$, so this is not a practical method for finding the solution $w^*(n)$, but the following calculations help elucidate the limits of \sys effectiveness as a function of the data and the partitioning, captured in matrix $\mat{P}$, and it also helps us show that \sys will in fact converge to $w^*(n)$ in expectation.

For convenience, we will work with the loss function $L(n,\vec{w}) = \sum \exp{(-n \vec{w} \vec{p_i})}$, which uses the approximation $(1-x)\lesssim e^{-x}$. Note $M(n, \vec{w}) \lesssim L(n,\vec{w})$, even in the case where the approximation is not tight, bounding this approximate loss from above also bounds the real one. We will use $\ell_i = \exp{(-n \vec{w} \vec{p_i})}$, so that $L(\vec{w},\vec{p},n) = \sum \ell_i$. 

	\begin{align}
	\nabla_{\vec{w}} \ell = -n\ell \vec{p} \\
	\nabla^2_{\vec{w}} \ell = n^2\ell \vec{p}\vec{p^{t}}
	\end{align}
	
	Therefore,
	\begin{align}
	\nabla_{\vec{w}} L = \sum_i -n \ell_i \vec{p_i} \label{eq:grad} \\
	\nabla^2_{\vec{w}} L  = \sum_i n^2 \ell_i \vec{p_i}\vec{p_i^t} \label{eq:hessian}
	\end{align}
	
	We can express $L= \vec{1}\cdot \exp{(-n\mat{P}\vec{w})}$, 
	the gradient $\nabla_{\vec{w}} L = -n \mat{P^t} \exp{(-n\mat{P}\vec{w})}$ and $\nabla^2_{\vec{w}} L  = n^2 \mat{P^t}\text{diag}(\exp{(-n\mat{P}\vec{w})})\mat{P}$. \\
	
	We may want to work with the log loss $\log {L}$ instead of $L$ directly.
	Let $\vec{s} = \text{softmax}(-n\mat{P}\vec{w})$, a vector with an entry for each instances ($N$) entries.
	In that case $\nabla_{\vec{w}} \log L = -n\cdot\mat{P^t} \vec{s}$ and $\nabla^2_{\vec{w}} \log L = n^2 \mat{P^t} (\text{diag} (\vec{s}) - \vec{s^t}\cdot\vec{s})\mat{P}$

\begin{align}
	\vec{w^*}(n) = \arg \min_\vec{w} L(n,\vec{w})	
\end{align}

	Subject to $\vec{w}$ having non negative entries that sum to 1.
	Note that $n\vec{w}*{p_i}$ is linear on each of its inputs, and $\exp(-x)$ is convex on $x$ (ie, looks like a smile), so that the overall loss $L$ is a sum of convex $\ell$, and hence also convex, and such a $\vec{w}$ is guaranteed to exist, though it may not be unique.
	
	We are interested in two things: how does $L(n, \vec{w^*}(n))$ differ from $L(\vec{1}, n)$, where $\vec{1}$ is an equal weighting of all chunks. 
\fi

%\bibliographystyle{ACM-Reference-Format} %SIGMOD
% \bibliographystyle{abbrv} %VLDB
%\bibliography{refs.bib} 
% \balance VLDB

\bibliographystyle{IEEEtran}
\bibliography{IEEEabrv,refs}

\if{0}
\newpage

\section*{Response to Reviewer Comments}
We thank the reviewers for their constructive comments, which have helped us substantially improve our paper.  We have highlighted our major changes in \textcolor{blue}{blue}.

\subsection*{Metareview}
\label{sec:metareview}
\begin{enumerate}
\item {\em The experimental evaluation needs to be largely strengthened. The authors need to use significantly different datasets that are of much larger size for evaluation, and report additional tests per reviews.}.  We have added 3 more datasets from BlazeIt, used both in \blazeit\ and in other papers. These are substantially different in content, increasing the diversity of video in our evaluation: they are from fixed cameras rather than moving cameras (dashcams in the original datasets), they are larger in size (30 hours each), and they are made up of a single camera stream rather than independent clips. We test a wide range of queries on them, bringing the total number of queries to 35. Additionally, we added extensive simulation experiments (Section \ref{sec:analysis}) designed to highlight properties of the datasets that allow \sys to work well. 
\item {\em Substantial improvements on writing are needed.} We apologize, and have revised our writing to make it easier to read. \srm{Can we say anything concrete about what we have done.}
\item {\em Additional discussions on related work are needed.} We expanded our discussion of related work to include suggestions from the reviewers, including more adjacent work from both the systems community and the computer vision community in a greatly expanded related work section.
\item \label{itm:analysis} {\em More discussions and in-depth analysis of the proposed approach are needed, including when the system may fail and other issues pointed out in the reviews.} In addition to expanding our evaluation, we greatly expanded our discussion of ExSample's upsides, trade-offs, and limitations through a combination of simulations over synthetic datasets and mathematical arguments in a new section (Section \ref{sec:analysis}).
\end{enumerate}

\subsection*{Reviewer 1}
\subsubsection*{Weakpoints and Detailed Evaluation}
\begin{enumerate}
\item {\em The main shortcoming to me is the evaluation. It's done on very limited datasets, two of the exact same kind. Little can be inferred from this. They are also very small.}. We have added more datasets in Section \ref{sec:newdatasets}. We also explain why the  existing datasets are more different than they may appear, because of the effect of having very many small video files in BDD, sampled from a large diversity of different driving scenarios, vs. the dashcam dataset which is a larger amount of video from a single car. The effect of having many files is explored through simulation in a new section Section \ref{sec:chunks}. We have also added a new section on simulation based-analysis, which is in Section \ref{sec:simulation}, which helps to shed light on when \sys performs well versus random sampling and when it does not.
\item {\em Further, the presentation lacks a critical analysis of the approach. Where does the approach, i.e., under what circumstances, does the approach fail? Using entirely different datasets (or somehow synthetic datasets) could shed a light on this.} We have added experiments on synthetic datasets to better characterize this.  We have also provided theoretical upper bounds on how well it performs. See \autoref{itm:analysis} in Section \ref{sec:metareview}.
\item {\em Another issue which I did not understand well, is how the $p$'s are defined for different objects? Does the user ultimately need to provide them? How are they determined? How does speed influence $p$? Again, using pretty much the same kind of video for this is not helpful at all. What $p$'s were used for the experiments? Will they be significantly different for other objects?} The user does {\em not} need to provide $p$, and \sys does not assume prior knowledge of these values. We wholly agree with your opinion that the usefulness of \sys would be much reduced otherwise: type of object, speed of object or camera motion speed would invalidate assumptions. We have updated our writing to make this clear early on, and in Section \ref{sec:simulation} we vary the $p_i$ by several orders of magnitude, and show that \sys consistently performs well. Adaptability to diverse datasets is a key strength of \sys.

\item{\em Overall, I found the method quite simple. It's rather intuitive and, in terms of cleverness, marginally better than random sampling.}  
We agree that the \sys algorithm is simple; we believe that this is part of its appeal, and it is exciting to us that such a simple method can outperform competing methods published in VLDB/SIGMOD.  We believe our core contributions are 1) providing a strong mathematical foundation that shows why and when \sys will help, and why \sys will not make things worse (which is not true of competing methods that require an upfront scan of the full dataset, as our results show) and 2) not requiring heavy prior knowledge about the data (either in the form of user configuration, hard-coded assumptions in the model, or pre-trained specialized models),  nor strong assumptions about the nature of data such as whether the camera is fixed or if it moves, or if events last a fixed amount of time.

\end{enumerate}
\subsection*{Reviewer 2}
\subsubsection*{Weak points} 
\begin{enumerate}
\item {\em For distinct object query, the algorithm presented follows the routine of ``sample frame, detect object, matching check''. However, it is not clear whether there are some existing efforts in the computer vision field that are tailored for distinct object detection or dynamic object detection in video, instead of the static object detection used in the paper? For example, would motion detection or tracking technique more suitable for distinct object query so that we can avoid matching check?} We have added some discussion of object detection in video as part of our related work section. Existing efforts in the CV community for dynamic object detection in video emphasize improving accuracy of detection by leveraging use of multiple frames at the same time, or some kind of temporal context, for each prediction. These are not necessarily cheaper and they would be orthogonal to what we do. 
Stand-alone tracking is cheaper than detection, but requires a scan of the frames. In \autoref{tab:xput} we show scanning with a light weight per-frame algorithm is not necessarily cheap.  Hence, the rationale behind avoiding tracking is similar to the rationale behind avoiding scanning.
\item {\em It is not clear under which circumstances ExSample outperforms other algorithms.} 
\begin{enumerate}
\item {\em For instance, with large recall (in LIMIT), would sequential execution be better than ExSample, especially with dynamic object detection or motion detection in the video?}. The simulation section (Section \ref{sec:simulation}) shows more clearly \sys helps save frames at every level of recall, but may struggle most early on. This is also shown in real data in \autoref{tab:byrecall}. For extremely high recall,  if we have no constraints on minimum  duration of an event, then all methods would need to process most of all the frames. Whether this is done sequentially is no longer  important, because the dominant cost is the expensive classifier. 

\item {\em When there are objects that only appear for a very short period, any random sampling scheme is going to be hard to capture such objects. If the query asks for large recall (in LIMIT), sequential execution is preferred in this adversary case.}

Our simulation explores this case of objects with very short duration, and we better characterize the cases when \sys will and will not outperform random (or sequential).

\end{enumerate}

\item {\em How much is the novelty of Bias and variance calculation? What is the bias and variance for the original Good-turing estimator? Is it similar to the one in the paper?} 
The main source we looked at for the original GT estimator is \cite{gtbounds}. The main result in \cite{good-biometrika} as it relates to our use case is about the estimator having low bias. Neither of these papers cover the situation we face when sampling frames.  Our contribution is an estimator that works well with the use cases of \sys, and we extend prior work in 4 ways: (1) we verify we can adapt it to a scenario where multiple results can appear in one sample, (2) we relate the bias to statistics of the query and data ($\mu_p$ and $\sigma_p$) (3) we model the uncertainty as $n$ grows so we can use it in a bandit context, and (4) we verify we can cover the case where a single instance may appear in different chunks. We have expanded our presentation of this analysis to clarify this.
\item {\em In the experiments, I think it may not be fair to include the pre-processing time of BlazeIt – the pre-processing time can be amortized across different queries (not only limit-k distinct object query) and can even be transferred to a different dataset.} We presented two views of the results. One uses frames processed as a metric. This metric excludes all preprocessing, and is a proxy for time (all of ExSample and BlazeIt and \random\ will essentially sample frames from the dataset and run an expensive classifier). This is the view in the left panels of Figures \ref{fig:bdd} and \ref{fig:dashcam}. On the right panels we include pre-processing costs. The preprocessing cost is due to both training and scanning the whole dataset, but is dominated largely by scanning the dataset with a query-specific proxy model. We added \autoref{tab:xput} to make this clear. Because the scanning section is done per-query, it cannot be amortized across queries: for any new type of query, we need to score every frame after training a new specialized proxy model.  It's possible that the proxy model could be reused if the same class of object was being detected on new datasets, but it is unclear how well such proxys models will generalize if data is sufficiently different. 
%The training stage can be amortized if any data that arrives is assumed to be similar to old data. This could be true of static cameras. Datasets with multiple moving cameras change more, and so it is more of an open question how much one can amortize even the training while keeping the lightweight model lightweight. The results on the BDD dataset suggest a lightweight model is not always a given. We think all of these different scenarios are important depending on the use case, and no one metric can capture these trade offs.
%\favyen{I think the main point is that the query-specific full scan dominates the runtime.}
\item {\em Why not using the same dataset as that in BlazeIt?} We have now added three of the four datasets from BlazeIt in Section \ref{sec:newdatasets}. They were not publicly available as of the first version of this paper. 
\item {\em Also, more baselines -- maybe the other approaches like Noscope and PP[12] have better performance than BlazeIt for distinct object queries?} We originally compared against NoScope, but found BlazeIt provides better performance. The main issue is that NoScope is inherently sequential, so it is much more likely than BlazeIt to keep applying the expensive detector on the same object. Additionally, BlazeIt can decide where to sample next with a global view of where the highest scores are, whereas NoScope has to decide ahead of time on a threshold and may miss objects entirely due to this. When we tried adapting NoScope to our use case it started to look more and more like BlazeIt, and consistently performed worse.
\end{enumerate}
\subsubsection*{Detailed evaluation}
\begin{enumerate}
\item {\em For Thompson sampling, a natural alternative is to formulate it as a multi-arm bandit problem and use upper confidence bound (UCB) as the strategy. Have you considered this alternative and how would this perform compared to Thompson sampling?} Thank you for the recommendation. We evaluated this alternative, specifically Bayes-UCB\cite{bayesucb}, using the same Gamma distribution model derived in Section \ref{sec:thompson}, and we observed almost identical trajectories with those of TS in the simulation in Section \ref{sec:simulation}. Part of the reason may be that Section \ref{sec:upperbound} shows an exact upper bound for any weighting scheme based on chunks, which thus applies to both  Thompson sampling and Bayes-UCB variants, as well as any other alternatives, and in Section \ref{sec:simulation} we observe \sys converges to the optimal allocation fairly quickly in many cases, so that the gains from any alternative bandit strategy are bound to be small especially at higher recall levels. 
\item {\em Is the video stored on Disk? how is the time spent on random access to different frames?} We find that the expensive object detector (20 fps) dominates runtime, and the I/O and decoding cost for random access of video frames is negligible. We measured this on a 1 GPU 1 SSD setup, with keyframes inserted every 10 frames. Although access times are not significant, being able to sample frames in batches and update chunk statistics asynchronously for different batches means \sys can in principle exploit parallel I/O, as well as overlap parallel I/O with parallel compute. 
\item {\em It would be good to have S3.3.2 not only on the synthetic dataset but also on the real dataset in S4.} To help bridge the gap between simulation and real data, we added Section \ref{sec:analysis} and \autoref{fig:simplot} as well as \autoref{fig:chunks}, on simulated datasets.  Although these datasets are simulated, they include significant skew that shows that our analysis in S3 in fact holds in practice.  We also showed that our real world data can be characterized by the parameters of these new simulations to estimate the performance of \sys on real datasets.  We hope those changes address the underlying goal of this request.
\item {\em Lots of typos}. We apologize for that and thank you for your patience when reading through our typos in the previous manuscript.  We have   proof-read more carefully this time.
\end{enumerate}
\subsubsection*{Revision}
\begin{enumerate}
\item {\em Discussion on any existing efforts in the computer vision field that are tailored for distinct object detection or dynamic object detection in video, instead of the static object detection used in the paper.} We have expanded our related work (Section \ref{sec:relwork}) to discuss several related works in this area.
\item {\em Discussion on when ExSample may fail.} We added Section \ref{sec:analysis} to address this, and also three new datasets in Section \ref{sec:eval}.
\item {\em In the experiment, present the time of BlazeIt without pre-training time; and use datasets in BlazeIt.} We have added 3 of 6 datasets in BlazeIt. The  remaining essentially have only one class of query (boat) and in those datasets there are boats everywhere all the time, because some components of BlazeIt target counting queries or queries such as `return frames where there are more than 7 boats' in addition to distinct object queries. The left sides of \autoref{fig:dashcam} and \autoref{fig:bdd} show the total frames processed, which is a proxy for time without any pre-training or scanning, and in \autoref{tab:xput} we break down pre-training from scanning.  As noted in our response above, even if a new proxy does not need to be trained, running BlazeIt on any video requires the proxy model to be run on every frame, which is a significant time cost.
\end{enumerate}
\subsection*{Reviewer 3}
\subsubsection*{Detailed Evaluation and Revision}
\begin{enumerate}
\item {\em The paper is not very well written, neither well explained. It’s difficult to follow the idea of the authors.} We apologize for the lack of clarity in some sections. We have used your comments to improve the writing, and believe you have helped us make the exposition clearer this time around.
\item {\em There are some format errors. Please proofread the manuscript carefully. For example, (i) In page 1, the citation item in the sentence ‘… as the rental price of a GPU is around \$3 per hour[?] is broken. (ii) In page 12, the content exceeds the page margin. (iii) Fig. 1 is blank.} We apologize for these and thank you for your patience with these issues while reviewing our manuscript. We have been more careful with proofreading this new version.
\item {\em This paper studies the problem of querying on large video. However, some important related works (e.g., [1,2]) in this area are missing.
[1] K. Hsieh et al., Focus: Querying large video datasets with low latency and low cost. OSDI18.
[2] J. Jiang et al., Chameleon: scalable adaptation of video analytics. SIGCOMM18.} We have added a discussion of these two papers in our related work section, as well as other related work in the data management, mobile systems, and computer vision communities. Thank you for bringing them to our attention.
\end{enumerate}
\fi

\end{document}